\documentclass[11pt]{article}

\usepackage[paperwidth=8.5in,paperheight=11in,top=1.25in, bottom=1.25in, left=0.75in, right=0.75in]{geometry} 

\usepackage{microtype}						 
\usepackage{natbib}
\usepackage{amsmath, amssymb, amsthm, amsfonts, cases}
\usepackage{mathrsfs}
\usepackage{url}
\usepackage{authblk}
\usepackage{bm}

\usepackage{graphicx}
\usepackage{epstopdf}
\usepackage[dvipsnames]{xcolor}
\usepackage[plainpages=false, pdfpagelabels]{hyperref} 
\hypersetup{
	colorlinks   = true,
	citecolor    = RoyalBlue,
	linkcolor    = red,
	urlcolor     = magenta
}
\usepackage{mathtools}                      
\mathtoolsset{showonlyrefs=true}          
\allowdisplaybreaks

\usepackage{todonotes}
\usepackage{float}

\newtheorem{rmk}{Remark}[section]
\newtheorem{lemma}{Lemma}[section]
\newtheorem{theorem}{Theorem}[section]

\newtheorem{corollary}{Corollary}[section]
\newtheorem{definition}{Definition}[section]
\newtheorem{ass}{Assumption}[section]
\newtheorem{pro}{Problem}[section]

\numberwithin{equation}{section}

\DeclareMathOperator*{\argmax}{arg\,max}

\newcommand\sig{\sigma}

\newcommand\lam{\lambda}

\newcommand\Del{\Delta}

\newcommand\Ac{\mathcal{A}}

\newcommand\Fc{\mathcal{F}}

\newcommand\Hc{\mathcal{H}}

\newcommand\Jc{\mathcal{J}}

\newcommand\Uc{\mathcal{U}}
\newcommand\Vc{\mathcal{V}}

\newcommand\Zc{\mathcal{Z}}

\newcommand{\Eb}{\mathbb{E}}
\newcommand{\Fb}{\mathbb{F}}
\newcommand{\Ib}{\mathbb{I}}
\newcommand{\Kb}{\mathbb{K}}
\newcommand{\Mb}{\mathbb{M}}
\newcommand{\Nb}{\mathbb{N}}
\newcommand{\Pb}{\mathbb{P}}

\newcommand{\Rb}{\mathbb{R}}

\newcommand{\dd}{\mathrm{d}}
\newcommand{\wn}{\widetilde{N}}

\newcommand{\wx}{\widehat{X}}
\newcommand{\wz}{\widehat{\Zc}}

\begin{document}

\title{Mean-Variance Portfolio Selection in Contagious Markets}

\author{Yang Shen\thanks{School of Risk and Actuarial Studies and CEPAR, University of New South Wales, Sydney, NSW 2052, Australia. Email:
\url{y.shen@unsw.edu.au}} \qquad
Bin Zou\thanks{Department of Mathematics, University of Connecticut, Storrs, CT, 06269-1069, USA.
Email: \url{bin.zou@uconn.edu}}}

\date{First Version: February 19, 2020\\This Version: \today \\ Accepted to \emph{SIAM Journal on Financial Mathematics}}
\maketitle

\begin{abstract}
\noindent
We consider a mean-variance portfolio selection problem in a financial market with contagion risk.
The risky assets follow a jump-diffusion model, in which jumps are driven by a multivariate Hawkes process
with mutual-excitation effect.
The mutual-excitation feature of the Hawkes process captures the contagion risk in the sense that each price jump of an asset increases the likelihood of future jumps not only in the same asset but also in other assets.
We apply the stochastic maximum principle, backward stochastic differential equation theory, and linear-quadratic control technique to solve the problem and obtain the efficient strategy and efficient frontier in semi-closed form, subject to a non-local partial differential equation.
Numerical examples are provided to illustrate our results.
\end{abstract}



\section{Introduction}
\label{sec:intro}

Asset prices exhibit jumps, occasionally and persistently, in all financial markets across the world, which has been well documented and empirically tested in the literature.
Large price movements are unlikely to be observed under standard financial models driven solely by Brownian motion(s), e.g., the Black-Scholes model.
The most popular models incorporating jumps are the jump-diffusion models, with the jump part predominantly driven by a Poisson process or a more general L\'evy process.
Those models have enjoyed great popularity in option pricing, term structure and credit risk modelling, and other applications.
One may refer to the survey article \cite{kou2007jump}, the monographs \cite{tankov2003financial} and \cite{oksendal2005applied}, and the references therein for detailed discussions of those models.

More strikingly, empirical studies confirm that the price jumps of an asset (or a class of assets) are likely to be accompanied by  more jumps, over a short time period, from not only the same (class of) asset  but also different (classes of) assets, creating a contagion or clustering effect in the market (see Figure \ref{fig:stock} for an example).
This contagion effect is observed across different markets globally, in particular, during financial crises.
The cascade of market declines experienced during the global financial crisis of 2007-2008 is a prominent example (see \cite{ait2015modeling} and \cite{azizpour2018exploring} for more examples).
As pointed out in \cite{ait2015modeling}, jumps under a standard L\'evy  jump-diffusion model are rare events, and, as a result, the clustering of large price movements cannot be properly explained using a standard jump-diffusion model.
The interplay between jumps from different assets or markets is complex, leaving the \emph{independent} increment assumption of L\'evy processes vulnerable.

\begin{figure}[h]
\begin{center}
\includegraphics[width = \textwidth]{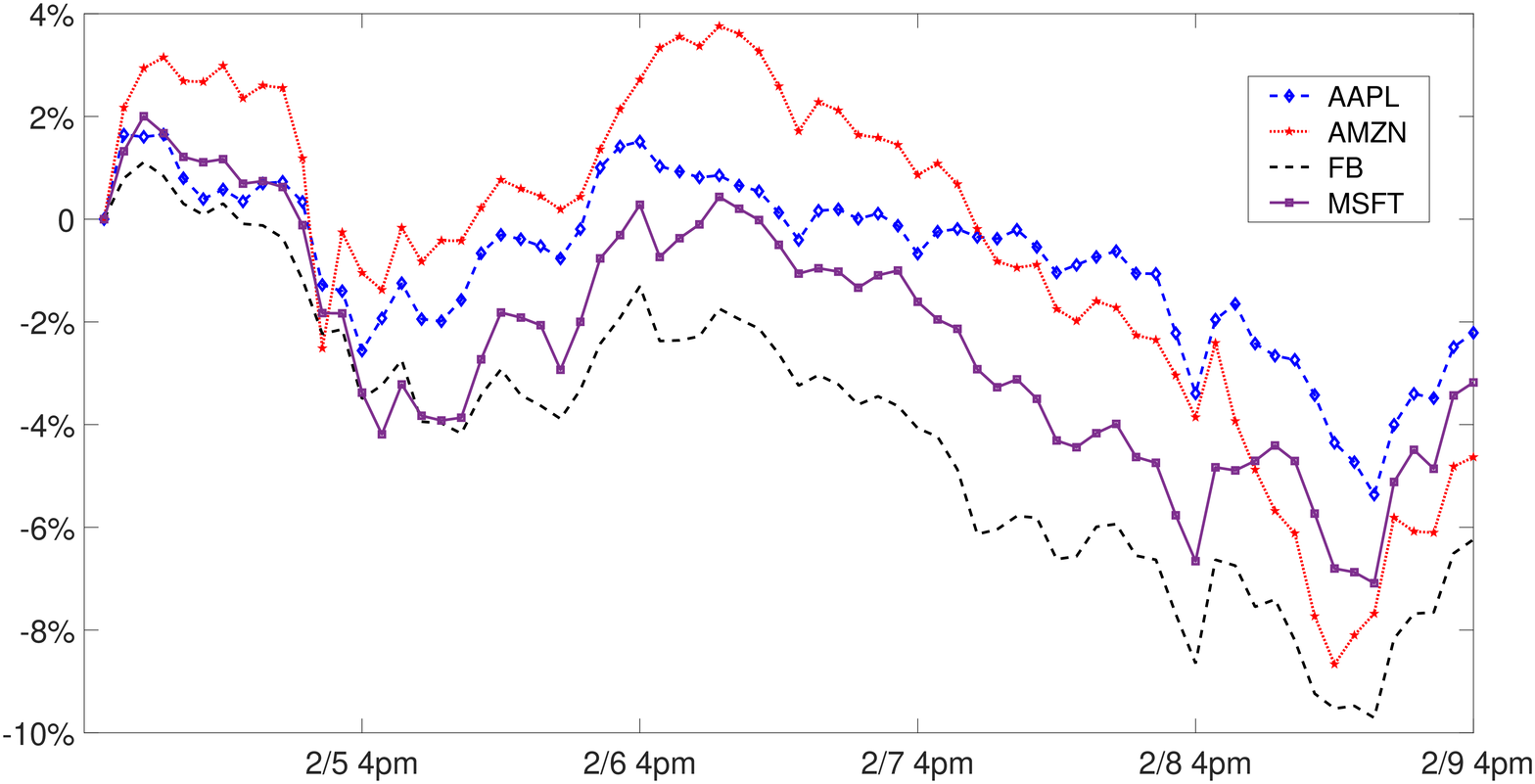}
\\[-2ex]
\label{fig:stock}
\end{center}
\begin{minipage}{\textwidth}
	\footnotesize
Figure 1. We plot the price changes in \emph{percentages} for AAPL (Apple), AMZN (Amazon), FB (Facebook) and MSFT (Microsoft) from 2/5/2018 9:30 am to 2/9/2018 4:00 pm during the 2018 February sell-off in the U.S. stock market.
The price data are in frequency 30m (minutes), and are taken to be the average of the highest and the lowest in each 30m window. Data are downloaded from Yahoo!Finance.
\end{minipage}
\end{figure}

Numerous approaches and models have been proposed to account for the clustering effect of jumps (extreme events) in finance and beyond.
Among them includes a successful candidate, \emph{Hawkes processes}, first introduced by \cite{hawkes1971spectra}.
In a multivariate Hawkes process $N=(N_1, N_2,\ldots,N_m)^\top$, the occurrence of a jump from one component (say $N_1$) raises not only the jump intensity of itself (self-excitation feature)
but also those of other components $N_2, \ldots, N_m$ (cross-excitation feature).
The joint effect of self-excitation and cross-excitation, called \emph{mutual-excitation}, increases the likelihood of seeing more jumps both over small time period and from different components, which provides a convincing explanation to the clustering of jumps observed in the financial markets.
An early paper using Hawkes processes to model financial data is \cite{bowsher2007modelling}. (Its working paper version first appeared online in 2002.)
More recent works that apply Hawkes processes to financial modeling include \cite{embrechts2011multivariate}, \cite{chavez2012high}, \cite{ait2015modeling}, and \cite{azizpour2018exploring}.
Hawkes processes have also been applied to risk analysis, credit risk, derivatives pricing, and other financial topics. See \cite{chavez2005estimating}, \cite{errais2010affine}, \cite{dassios2011dynamic}, \cite{ait2014mutual},
and \cite{zhang2018bond}, among many others.
In actuarial science, \cite{dassios2012ruin} and \cite{zhu2013ruin} analyze ruin probabilities when the arrivals of claims are modeled by a Hawkes process.
We refer to the survey papers of \cite{bacry2015hawkes} and \cite{hawkes2018hawkes} for the overview and progress of Hawkes processes in finance.

In this paper, we consider the classical mean-variance (MV) portfolio selection problem (pioneered by \cite{markowitz1952portfolio}) in a financial market with contagion risk.
The asset prices are modeled by a multi-dimensional jump-diffusion model, in which self-excited and cross-excited jumps
are driven by a multivariate Hawkes process.
The mutual-excitation feature of the Hawkes process captures the clustering and contagion effects in the financial market.
It is well known that the MV portfolio selection problem is a \emph{time inconsistent} control problem, in the sense that an optimal strategy found at the current time may cease to be optimal as time evolves. Two lines of research are developed to tackle such a problem.
The first method is to treat the problem as a \emph{precommitted} problem and solve for an optimal strategy; see \cite{bajeux1998dynamic}, \cite{li2000optimal}, \cite{zhou2000continuous},  \cite{lim2002mean}, \cite{zhou2003markowitz}, and \cite{yin2004markowitz}, among many others.
The second approach is to treat the problem under a game theoretical framework and seek an equilibrium strategy; see for instance \cite{basak2010dynamic} and \cite{bjork2014mean}.
\cite{shen2021mean} provide detailed comparisons between these two approaches.
We adopt the first precommitted framework in this paper, namely, we assume that once an optimal strategy is obtained, the agent will commit herself to this strategy throughout the entire investment horizon.

Although Hawkes processes have enjoyed great success in financial modeling, studies on control problems with systems driven by Hawkes processes are still rare, especially when comparing to those driven by Brownian motions and/or
Poisson jump processes.
To our knowledge, \cite{ait2016portfolio} is the earliest paper which considers the portfolio selection problem
in a contagious market driven by a Hawkes process.
Another related paper is \cite{bo2019credit}, which solves a portfolio optimization problem under a defaultable market
modeled by mean-reverting diffusion processes enhanced with self-excitation.
\cite{cao2019optimal} apply a Hawkes process to model the claim frequency for an insurer and seek the optimal strategy to the insurer's optimal reinsurance problem.
\cite{liu2021household} study an optimal investment, consumption, and life insurance problem in a self-contagious market (i.e., there is only one risky asset) for a utility-maximizing household with a bequest motive.

We summarize the key contributions and findings of this paper as follows:
\begin{itemize}
	\item This paper contributes to the stochastic control literature when the controlled system is driven by a Brownian motion and a Hawkes process. In particular, in the area of portfolio selection problems, ours shall follow the work of \cite{ait2016portfolio}
	by modeling the contagion risk via a Hawkes process,
	but with significant differences.
	The objective in \cite{ait2016portfolio} is to maximize the expected utility of consumption over an infinite time horizon, while we consider the MV problem over a finite horizon. The difference in methodology is further elaborated below.
	
	\item In all the related papers mentioned above, the standard dynamic programming principle is used and the optimal strategy is found by solving the associated Hamilton-Jacobi-Bellman (HJB) equation. In comparison, we apply the stochastic maximum principle,
	backward stochastic differential equation (BSDE) theory, and
	linear-quadratic (LQ) control technique to solve our MV portfolio selection problem. Given that the stochastic maximum principle can be applied to deal with non-Markovian\footnote{The control systems are (joint) Markovian in \cite{ait2016portfolio}, \cite{cao2019optimal}, and \cite{liu2021household}, while \cite{gao2018dp} study a self-exciting non-Markovian Hawkes process.} control systems with random coefficients, our approach has the potential to tackle portfolio selection problems under a more general setup with non-Markovian Hawkes processes (see \cite{gao2018dp}).
	
\item Numerical findings: In the numerical analysis of a univariate example, we find that as the initial
intensity value, mean-reversion level, or jump size of the intensity process increases, the efficient frontier deteriorates and the MV  investor is worse off. On the contrary, a larger mean-reversion speed improves the efficient frontier. Compared
with the Poisson-jump-diffusion model with the same initial intensity, the Hawkes model may lead to either an improved or deteriorated efficient frontier, depending on the relative size of the initial intensity and the mean-reversion level in the Hawkes model.
\end{itemize}

The rest of the article is organized as follows. In Section \ref{sec:prob}, we introduce a contagious financial market and present the MV portfolio selection problem.
In Section \ref{sec:deri}, we carry out heuristic derivations to solve an equivalent problem using the stochastic maximum principle and BSDE theory.
In Section \ref{sec:op}, we apply the completing the square technique in LQ theory and formally obtain the efficient (optimal) strategy and the efficient frontier of the MV problem.
We conduct numerical analysis in Section \ref{sec:exm} and summarize concluding remarks in Section \ref{sec:con}. Technical proofs are collected in Appendixes A-C. Particularly, Appendix C is devoted to
the discussion of the connection between the stochastic maximum principle and the HJB equation approach.

\section{Problem formulation}
\label{sec:prob}

In this section, we first introduce a multi-dimensional Hawkes process to model mutually exciting jumps, which capture the contagion risk in the financial market. Then, we
formulate the mean-variance portfolio selection problem
in \eqref{pro:mv}.

\subsection{Probabilistic setup and Hawkes jump modeling}
\label{subsec:jump}

Let $[0, T]$ be a finite time horizon, where $T < \infty$ represents the terminal time of planning.
We consider a complete probability space $(\Omega, {\cal F}, {\mathbb P})$
equipped with a filtration ${\mathbb F} : = \{ {\cal F}_t \}_{t \in [0, T]}$, which is assumed to carry all the random objects considered in the sequel and satisfy the usual conditions of right-continuity and ${\mathbb P}$-completeness.
The operator $\Eb$ (resp. $\mathrm{Var}$) denotes taking expectation (resp. variance) under $\Pb$, and $\Eb_t := \Eb[\cdot | \Fc_t]$.
On this probability space, two stochastic processes are defined: an $n$-dimensional standard Brownian motion $W: = \{ (W_{1} (t), W_{2} (t), \ldots, W_{n} (t) )^\top \}_{t \in [0, T]}$ and an $m$-dimensional c\`adl\`ag (right continuous with left limits) point process $N := \{ (N_{1} (t), N_{2} (t), \ldots, N_{m} (t))^\top \}_{t \in [0, T]}$. Here, $n$ and $m$ are two positive integers, and ${}^\top$ denotes the usual transpose operator on a vector or a matrix.
Let $k$ be another positive integer, we define three index sets by\footnote{Set $\Nb$ is the index set of $n$ components of the Brownian motion $W$, set $\Mb$ is the index set of $m$ components of the Hawkes process $N$, and $\Kb$ is the index set of $k$ risky assets (introduced in the next subsection).}
\begin{align}
\label{eq:set}
\Nb := \{1,2,\ldots,n\}, \qquad \Mb:=\{1,2,\ldots,m\}, \qquad \text{and} \qquad \Kb:=\{1,2,\ldots,k\}.
\end{align}

Let us denote the intensity process of $N$ by $\lam : = \{ (\lambda_{1} (t), \lambda_{2} (t), \ldots,
\lambda_{m} (t))^\top \}_{t \in [0, T]}$, where
$\lam_l(t)$ is the corresponding instantaneous intensity of $N_l$ at time $t$ for each $l \in \Mb$ and $t \in [0,T]$.
Hereinafter, we work with the c\`agl\`ad version of $\lam$, i.e., $\lam$ is left continuous with right limits.
Heuristically, we have (for rigorous definitions see \cite{daley2007introduction})
\begin{align*}
\lam(t) = \lim_{\Del t \to 0} \frac{\Eb_t \left[ N(t+\Del t) - N(t) \right]}{\Del t}, \qquad t \in [0,T].
\end{align*}
In addition, we define a new process $\wn: = \{ ({\widetilde N}_{1} (t),{\widetilde N}_{2} (t), \ldots, {\widetilde N}_{m} (t))^\top \}_{t \in [0, T]}$,  commonly referred to as the compensated process of $N$, by
\begin{align*}
\wn(t) := N(t) - \int_0^t \lam(u) \, \dd u, \qquad t \in [0,T].
\end{align*}
It is well known that $\wn$ is an $m$-dimensional $(\Fb, \Pb)$-local martingale.

In this paper, we use a Hawkes process, introduced in \cite{hawkes1971spectra}, to model the point process $N$.
In particular, we assume that the intensity process $\lam_l$ of $N_l$ is governed by
\begin{align}
\dd  \lambda_l (t) &= \alpha_l ( \lambda_{l \infty} - \lambda_l (t) ) \, \dd t
+ \sum^m_{j = 1} \beta_{lj} \, \dd N_j (t)  \notag \\
&= \bigg [ \alpha_l ( \lambda_{l \infty} - \lambda_l (t) ) + \sum^m_{j = 1} \beta_{lj} \lambda_j (t) \bigg ] \dd t
+ \sum^m_{j = 1} \beta_{lj} \, \dd {\widetilde N}_j (t), \label{eq:dlam_l}
\end{align}
where $\lambda_l (0) = \lambda_{l0}$ and $l \in \Mb$.
We assume that all the coefficients in \eqref{eq:dlam_l} are constants, and $\lambda_{l0} > 0$,
$\alpha_l > 0$, $\lambda_{l \infty} \geq 0$ and $\beta_{lj} \geq 0$, for all $l, j \in \Mb$.
The dynamics equation given by \eqref{eq:dlam_l} implies that the intensity process $\lam_l$ mean reverts to the long-term rate $\lam_{l \infty}$ at speed $\alpha_l$, and jumps up by size $\beta_{lj}$ whenever process $N_j$ jumps, where $j \in \Mb$.
Alternatively, we can solve \eqref{eq:dlam_l} and represent $\lam_l$ by
\begin{align}
\lam_l(t) = e^{-\alpha_l t} \, \lam_{l0} + \left(1 - e^{-\alpha_l t}\right) \, \lam_{l \infty} + \sum_{j=1}^m \int_0^{t-} \, \beta_{lj} e^{-\alpha_l (t-s)} \dd N_j(s), \qquad t\in[0,T], \, l\in \Mb.
\end{align}
Obviously, $N$ is an $m$-dimensional Hawkes process with \emph{exponential} decay.

With $\lam_l$ satisfying \eqref{eq:dlam_l} for all $l \in \Mb$, we write the dynamics of $\lam$ in the following vector form
\begin{align}\label{eq:intensity}
\dd \lambda (t) = \alpha (\lam_\infty - \lam(t)) \, \dd t + \beta \, \dd N(t)
=\big ( \alpha \lambda_\infty + (\beta - \alpha) \lambda (t) \big ) \, \dd t + \beta \, \dd {\widetilde N} (t) ,
\end{align}
where we denote
\begin{align}
\label{eq:vec}
\begin{split}
\lambda (0) = \lambda_0 : = (\lambda_{10}, \lambda_{20}, \ldots, \lambda_{m0})^\top,& \qquad
\lambda_{\infty} : = (\lambda_{1 \infty}, \lambda_{2 \infty}, \dots, \lambda_{m \infty})^\top, \\
\alpha := \mbox{Diag} [(\alpha_1, \alpha_2, \ldots, \alpha_m)^\top],& \qquad
\beta := [\beta_{lj}]_{l, j \in \Mb}.
\end{split}
\end{align}
Here, for a vector $v$, $\mathrm{Diag}[v]$ denotes the square matrix, whose diagonal vector is equal to $v$ and other entries are zero.

\begin{rmk}
As can be easily seen from \eqref{eq:intensity}, a jump of $N_l$ at time $t$ not only increases its own instantaneous intensity $\lam_l(t)$ by $\beta_{ll}$ (\emph{self-excitation}) but also increases the instantaneous intensities $\lam_j(t)$ of other processes $N_j$ by $\beta_{jl}$ (\emph{cross-excitation}), where $j \neq l, j\in \Mb$. Furthermore, a jump occurring at time $t$ has a permanent impact on the intensity process $\lam_l(s)$ for all $s \ge t$ and $l \in \Mb$, although such an impact decays exponentially at rate $\alpha_l$.
As a result, the intensity process $\lam$ depends on the past history of $N$, making $N$ a path-dependent process.
The dynamics of $\lam$ in \eqref{eq:intensity} indicate that $(N, \lam)$ is a joint Markov process.
The above Hawkes jump model is also used in \cite{ait2016portfolio}, \cite{cao2019optimal}, and \cite{liu2021household}.
\end{rmk}

\subsection{Financial market with contagion risk}

We consider a financial market that consists of one risk-free asset (e.g., savings account) and $k$ risky assets (e.g., stocks). The risk-free asset earns interest continuously at a constant rate $r>0$, and its price evolves according to
\begin{align}
\dd S_{0} (t) = r S_{0} (t) \, \dd t , \qquad S_0 (0) = 1.
\end{align}
The price processes of the risky assets are given by
\begin{align}
\label{eq:dS}
\dd S_{i} (t) = S_{i} (t-) \bigg [ \mu_i \, \dd t + \sum^n_{j = 1} \sigma_{ij} \, \dd W_{j} (t)
+ \sum^m_{l = 1} J_{il} \Big ( Z_{l} (t) \, \dd N_{l} (t) - \Eb [Z_l(t)] \lambda_l (t) \, \dd t \Big ) \bigg ], \qquad i \in \Kb,
\end{align}
where $S_i (0) = s_{i0} > 0$ and $N_l$ is the $l^{th}$ entry of the $m$-dimensional Hawkes process $N$
introduced in Section \ref{subsec:jump}, for $l \in \Mb$. Recall set $\Kb=\{1,2,\ldots,k\}$ (see \eqref{eq:set}).
The assumptions of model \eqref{eq:dS} will be presented shortly after a brief introduction of notations.
Our model \eqref{eq:dS} is similar to that of \cite{ait2016portfolio} (see Eq.(2.2) in their paper), and both consider exponential decay for the intensity process (see Eq.\eqref{eq:dlam_l} or \eqref{eq:intensity} in ours and Eq.(2.3) in theirs).

For notational simplicity, we introduce the expected return rate vector $\mu$,
the risk premium vector $B$,
the volatility matrix (of dimension $k \times n$) $\sigma$, where
\begin{align}
\label{eq:notation_vec}
\mu : = ( \mu_1 , \mu_2 , \ldots, \mu_k )^\top, \quad B : = ( \mu_1 - r, \mu_2 - r, \ldots, \mu_k - r )^\top, \quad \sigma : = [ \sigma_{ij} ]_{i \in \Kb, \, j \in \Nb},
\end{align}
 and the jump size matrix (of dimension $k \times m$) $\eta$, where
\begin{align}
\label{eq:eta}
\eta(Z (t)) : = [ J_{il} Z_{l} (t) ]_{i \in \Kb, \, l \in \Mb}, \qquad \forall \, t \in [0,T].
\end{align}
We summarize below the model assumptions that will be imposed throughout the paper.

\begin{ass}
\label{ass:market}
In the market model \eqref{eq:dS}, the drift $\mu_i>0$, the volatility rate $\sig_{ij} >0$, and the scaling factor $J_{il} \in [0,1]$ are all constants, for all $i \in \Kb$, $j \in \Nb$, and $l \in \Mb$. The jump size $Z_l = \{Z_l(t)\}_{t \in [0,T]}$ is a series of independent and identically distributed (i.i.d.) random variables, with common probability measure $\nu_l$ supported on $(-1, \infty)$ and finite second moment, for all $l \in \Mb$.
Let us denote $\nu:=(\nu_1,\nu_2,\ldots,\nu_l)^\top$.
We assume the following non-degenerate condition holds true for both the variance-covariance matrix $\sigma \sigma^\top$
and the precision matrix $(\sigma \sigma^\top)^{-1}$ (see for instance \cite{zhou2003markowitz} for similar assumptions):
 \begin{align}
 \zeta^\top \, \sigma \sigma^\top \zeta\geq \epsilon |\zeta|^2 , \quad \zeta^\top \, (\sigma \sigma^\top)^{-1} \zeta \geq \epsilon |\zeta|^2 , \qquad \forall \, \zeta \in \mathbb{R}^k ,
 \end{align}
 where $\epsilon$ is a positive constant.
 The jumps in asset prices are modeled by the $m$-dimensional Hawkes process $N$, whose intensity process $\lam$ is given by \eqref{eq:intensity}.
 Furthermore, we suppose that the Brownian motion $W$, the Hawkes process $N$, and the jump size random variables $Z = \{Z_l\}_{l \in \Mb}$ are stochastically independent of each other.
 The filtration $\Fb$ is the augmented filtration generated by $W$, $N$, and $Z$.
\end{ass}

Let $u$ and $v$ be two arbitrary vectors with the same dimension (say $n$), we define operator $\bullet$ by
\begin{align}
\label{eq:bullet}
u \bullet v := (u_1 v_1, u_2v_2, \ldots, u_nv_n)^\top.
\end{align}

\begin{lemma}
	\label{lem:Sigma}
The generalized variance-covariance matrix, defined by
 \begin{align}
 \label{eq:Sigma}
 \Sigma (t) : = \sigma \sigma^\top
 + \int_{(-1, \infty)^m} \eta (z) \, \mathrm{Diag}[\lam (t) \bullet \nu (\dd z)] \, \eta (z)^\top
 \end{align}
is non-degenerate and positive definite. Moreover, $\Sigma (t)^{-1}$ is positive semi-definite.
\end{lemma}

\begin{proof}
We first note that,	in the definition of \eqref{eq:Sigma}, $\sigma$ and $\eta$ are given by \eqref{eq:notation_vec} and \eqref{eq:eta}, $\nu$ is the distribution function of jump random variables $Z$, and $\lam$ is the intensity process, which is positive and governed by \eqref{eq:intensity}.
$\mathrm{Diag}[\lam (t) \bullet \nu (\dd z)]$ denotes the square ($m \times m$) matrix with diagonal entries equal to  $\lam (t) \bullet \nu (\dd z)$, where
\[ \lam (t) \bullet \nu (\dd z) = \big ( \lam_{1} (t) \nu_1 (\dd z_1), \lam_{2} (t) \nu_2 (\dd z_2),
\ldots, \lam_{m} (t) \nu_m (\dd z_m) \big )^\top , \]
and all other entries are equal to 0.	

By simple algebra, we obtain, for any $\zeta \in \mathbb{R}^k$, that
\begin{align}
\int_{(-1, \infty)^m} \zeta^\top \, \eta (z) \, \mathrm{Diag}[\lam (t) \bullet \nu (\dd z)] \, \eta (z)^\top \, \zeta
= \sum^m_{l = 1} \int_{(-1, \infty)} \big[(\zeta^\top \, \eta (z) )_l\big]^2 \, \lam_l (t) \, \nu_l (\dd z_l) \geq 0 ,
\end{align}
where $(\zeta^\top \eta (z))_{l}$ denotes the $l^{th}$ entry of the
$m$-dimensional vector $\zeta^\top \eta (z)$, for all $l \in \Mb$.
Thus, by the non-degenerate condition in Assumption \ref{ass:market}, we have
\begin{align}
\zeta^\top \Sigma (t) \zeta\geq \zeta^\top \sigma \sigma^\top \zeta
\geq \epsilon |\zeta|^2, \qquad \forall \, \zeta\in \mathbb{R}^k ,
\end{align}
which implies $\Sigma (t)$ is non-degenerate and positive definite.
The positive semi-definiteness of $\Sigma (t)^{-1}$ also holds true because
\begin{align}
\zeta^\top \Sigma (t)^{-1} \zeta= (\zeta^\top \Sigma (t)^{-1}) \, \Sigma (t) \, (\Sigma (t)^{-1} \zeta)
\geq \epsilon |\Sigma (t)^{-1} \zeta|^2 \geq 0 , \qquad \forall \, \zeta \in \mathbb{R}^k .
\end{align}
The proof is then completed.
\end{proof}

\begin{rmk}
The market model \eqref{eq:dS} naturally inherits both self-excitation and cross-excitation features from the Hawkes process $N$ (together called mutual-excitation). Such a financial market is contagious in the sense that, a jump from one source will increase the jump intensity of all the sources and then likely lead to more jumps from different sources in the near future, creating contagion effects and clustering phenomena. We also point out that, by imposing $J_{il} \in [0, 1]$ and $Z_l \in (-1, \infty)$, the asset price $S_i$ stays strictly positive for all $i \in \Kb$. \cite{ait2016portfolio} focus on negative jumps only and further assume $J_{il} Z_l \in (-1,0)$ for all $i \in \Kb$ and $l \in \Mb$.
\end{rmk}

Moreover, we follow the standard notation of random measures (see, e.g., Eq.(1.1.2) in \cite{oksendal2005applied}) and define (with slight abuse of notations)
\begin{align}
N_l(t, \Ib) := \sum_{0< s \le t} \chi_{\Ib}(\Del N_l(s)), \qquad \forall \, l \in \Mb, \, t \in [0,T],
\end{align}
where $\Ib \subset (-1, \infty)$, $0 \not\in \bar{\Ib}$ (the closure of $\Ib$), and $\chi$ is an indicator function. Let us define
\begin{align}
{\widetilde N}_l (\dd t, \dd z_l) &: = N_l (\dd t, \dd z_l) - \lambda_l (t) \, \nu_l (\dd z_l)\, \dd t , \qquad l \in \Mb,\\
\text{and} \qquad {\widetilde N} (\dd t, \dd z) &: = \left( {\widetilde N}_1 (\dd t, \dd z_1), {\widetilde N}_2 (\dd t, \dd z_2),  \label{eq:wt_N}
\ldots, {\widetilde N}_m (\dd t, \dd z_m) \right)^\top .
\end{align}
By using the random measure notation, we naturally have
\begin{align}
\int^t_0 J_{il} \big ( Z_{l} (t) \, \dd N_{l} (t) - \Eb [Z_l(t)] \lambda_l (t) \, \dd t \big )
=  \int^t_0 \int_{(-1,\infty)} J_{il} z_l  \, {\widetilde N}_l (\dd t, \dd z_l), \qquad \forall \, t\in[0,T], \, l \in \Mb .
\end{align}
We can also rewrite the market model \eqref{eq:dS} in the following vector form
\begin{align}
\label{eq:dS_vec}
\dd S(t) = \mathrm{Diag}[S(t-)] \left( \mu \, \dd t + \sig \, \dd W(t) +  \int_{(-1,\infty)^m} \; \eta(z) \, \wn(\dd t, \dd z) \right).
\end{align}

We consider a representative agent (investor), who is a price-taker in the economy (i.e., her trading does not move the prices of the assets).
In the financial market with contagion risk, as described above, the agent chooses her investment strategy $\pi : = \{ (\pi_1 (t), \pi_2 (t), \ldots, \pi_k (t))^\top \}_{t \in [0, T]}$, where $\pi_i (t)$ represents the \emph{dollar amount} invested in the $i^{th}$ risky asset at time $t$.
Let us denote by  $X^\pi := \{ X^\pi (t) \}_{t \in [0,T]}$ the wealth process associated with strategy $\pi$. We consider self-financing strategies only, that means the remaining amount of $X^\pi(t) - \sum_{i=1}^k \pi_i(t)$ is fully invested in the risk-free asset at time $t$.
We shall write $X (t) := X^\pi (t)$ for simplicity, whenever there is no risk of confusion.
Using \eqref{eq:dS_vec}, we obtain the dynamics of $X$ by
\begin{align}\label{eq:wealth}
\dd X (t) = \big ( r X (t) + \pi (t)^\top B \big ) \, \dd t + \pi (t)^\top \sigma \, \dd W (t)
+ \int_{(-1, \infty)^m} \pi (t)^\top \eta (z) {\widetilde N} (\dd t, \dd z) ,
\end{align}
where $X (0) = x_0 > 0 $ is the agent's initial wealth, the risk premium $B$ and volatility matrix $\sig$ are given by \eqref{eq:notation_vec}, and $\eta$ is the jump size matrix defined by \eqref{eq:eta}.

Before introducing the MV problem, we define three spaces of processes that will be used in the subsequent analysis.

\begin{definition}
\label{def:space}
Denote by ${\cal P}$ the ${\mathbb F}$-predictable $\sigma$-field on $\Omega \times [0, T]$.
We define
\begin{itemize}
	\item ${\cal L}^2_{\cal F} (0, T; {\mathbb R}^k)$ as the space of all ${\mathbb R}^k$-valued, ${\mathbb F}$-predictable processes $\varphi = \{ \varphi (t) \}_{t \in [0, T]}$ such that
	\begin{align*}
	{\mathbb E} \bigg [ \int^T_0 |\varphi (t)|^2 \, \dd t \bigg ] < \infty;
	\end{align*}
	
	\item ${\cal L}^{2, N}_{\cal F} (0, T; {\mathbb R}^m)$ as the space of all ${\mathbb R}^m$-valued,
	${\cal P} \times {\cal B} ((-1, \infty)^m)$-measurable processes $\varphi = \{ \varphi (t, z) \}$
	such that
	\begin{align*}
	{\mathbb E} \bigg [ \int^T_0 |\varphi (t, \cdot)|^2_N \, \dd t \bigg ] < \infty ,
	\end{align*}
	where $\cal B$ denotes the Borel set, $\varphi (t, z) : = (\varphi_1 (t, z_1), \varphi_2 (t, z_1), \ldots, \varphi_m (t, z_m))^\top$, with $(t, z) \in [0, T] 	 \times (-1, \infty)^m$,  and
	\begin{align*}
	|\varphi (t, \cdot)|^2_N : = \sum^m_{l = 1} \int_{(-1,\infty)} |\varphi_l (t, z_l)|^2
	\lambda_l (t) \nu_l (\dd z_l);
	\end{align*}
	
	\item ${\cal S}^2_{\cal F} (0, T; {\mathbb R})$ as the space of all real-valued, ${\mathbb F}$-adapted,
	c\`adl\`ag processes $\varphi = \{ \varphi (t) \}_{t \in [0, T]}$ such that
	\begin{align*}
	{\mathbb E} \bigg [ \sup\limits_{t \in [0, T]} |\varphi (t)|^2 \bigg ] < \infty .
	\end{align*}
\end{itemize}
\end{definition}

As is standard in the literature (see Definition 2.1 in \cite{zhou2003markowitz}), we need to impose certain integrability conditions on the investment strategy $\pi$, which are listed in the definition below.

\begin{definition}
	\label{def:adm}
An investment strategy $\pi$ is said to be admissible if $\pi$ is self-financing,
$\pi \in {\cal L}^2_{\cal F} (0, T; {\mathbb R}^k)$
and $\eta^\top \pi  \in {\cal L}^{2, N}_{\cal F} (0, T; {\mathbb R}^m)$, where spaces ${\cal L}^2_{\cal F} (0, T; {\mathbb R}^k)$ and ${\cal L}^{2, N}_{\cal F} (0, T; {\mathbb R}^m)$ are defined in Definition \ref{def:space}.
Denote by ${\cal A}$ the set of all admissible strategies.
\end{definition}

\begin{rmk}
For any $\pi \in {\cal A}$, the stochastic differential equation (SDE) of $X$ given by   \eqref{eq:wealth} admits
a unique strong solution $X$, which belongs to the space ${\cal S}^2_{\cal F} (0, T; {\mathbb R})$ (see Definition \ref{def:space}).
In consequence, we have $\Eb[X(T)] < \infty$ and
${\mathbb E} [ X^2 (T) ] < \infty$, implying that 
the variance $\mathrm{Var}[X(T)]$ is well defined.
However, for utility maximization problems, an additional requirement of $X$ being bounded from below (e.g., $X >0$ a.s.) is often imposed to exclude certain strategies (e.g., doubling strategies) and make the problems well posed (see \cite{karatzas1998methods}[Definition 3.2]).
\end{rmk}

We now proceed to formulate the MV portfolio selection problem that is of pivotal interest to this paper.

\begin{pro}
	\label{pro:MV}
The agent seeks to solve the following mean-variance portfolio selection problem:
\begin{align}\label{pro:mv}
\left\{
\begin{aligned}
& \min_{\pi  \in {\cal A}} \ J (x_0, \lambda_0; \pi)
:= \min_{\pi  \in {\cal A}} \ {\mathbb E} [(X(T)-\xi)^2] \\
& \mbox{subject to} \
\left\{
\begin{aligned}
& {\mathbb E} [ X (T) ] = \xi , \\
& X \text{ is the solution to } \eqref{eq:wealth},
\end{aligned}
\right.
\end{aligned}
\right.
\end{align}
where $x_0$ is the agent's initial wealth and $\lam_0$ is the initial value of the intensity process $\lam$ given by \eqref{eq:intensity}.
We call a solution $\pi^*$ to the above problem an optimal (investment) strategy or an efficient strategy.
To avoid trivial cases, we set $\xi \geq x_0 e^{rT}$.
\end{pro}

\begin{rmk}
The formulation of problem \eqref{pro:mv} is the same as that in \cite{lim2002mean} and \cite{zhou2003markowitz}.
In \cite{zhou2000continuous}, an alternative formulation of the MV portfolio selection problem is considered:
\begin{align}
\min\limits_{\pi\in{\cal A}} \, \left\{ - {\mathbb E} [ X (T) ] + \gamma \mathrm{Var} (X (T)) \right\}, \qquad \gamma >0 .
\end{align}
Due to the variance term, this alternative one is not a standard control problem, although it does not involve expectation constraint. A standard approach to overcoming the difficulty caused by the variance term is the embedding technique developed in \cite{li2000optimal}[Theorem 2] and \cite{zhou2000continuous}[Theorem 3.1].
\end{rmk}

Notice that problem \eqref{pro:mv} is a constrained stochastic optimization problem, with an equality constraint. To solve this problem, we consider a modified objective $J_1$, defined by
\begin{align}
J_1 (x_0, \lambda_0; \pi, \theta)
:&= {\mathbb E} \big [ (X (T) - \xi)^2 \big ] + 2 \theta \big ( {\mathbb E} [ X (T) ] - \xi \big )
= {\mathbb E} \big [ (X (T) - (\xi - \theta))^2 \big ] - \theta^2,
\end{align}
where $\theta \in \Rb$ is the Lagrange multiplier. Please refer to \cite{lim2002mean} and \cite{zhou2003markowitz} for similar ideas.
As discussed in these two papers, it follows from the Lagrangian duality theorem that solving the MV portfolio selection problem \eqref{pro:mv} is equivalent to solving the following max-min problem:
\begin{align}\label{pro:min-max}
\left\{
\begin{aligned}
& \max_{\theta \in {\mathbb R}} \min_{\pi \in {\cal A}} \ J_1 (x_0, \lambda_0; \pi, \theta)
= \max_{\theta \in {\mathbb R}} \min_{\pi \in {\cal A}} \left\{ {\mathbb E} \big [ (X (T) - (\xi-\theta))^2 \big ] - \theta^2 \right\} \\
& \text{subject to }
X  \text{ given by } \eqref{eq:wealth}.
\end{aligned}
\right.
\end{align}
To solve problem \eqref{pro:min-max}, we first solve the quadratic-loss minimization problem:
\begin{align}\label{pro:qlm}
\left\{
\begin{aligned}
& \min_{\pi \in {\cal A}} \ J_2 (x_0, \lambda_0; \pi , c )
:= \min_{\pi \in {\cal A}} \ {\mathbb E}
\big [ (X (T) - c)^2 \big ]  \\
& \text{subject to } X  \text{ given by } \eqref{eq:wealth},
\end{aligned}
\right.
\end{align}
where $c \in \Rb$ is a free parameter in problem \eqref{pro:qlm}.
Let us denote by $ \Jc_2(x_0, \lam_0; c)$ the value function of problem \eqref{pro:qlm}, i.e.,
$\Jc_2(x_0, \lam_0; c):= J_2(x_0, \lam_0; \pi^*_c, c)$,
where $\pi^*_c$ is an optimizer to problem \eqref{pro:qlm} associated with the free parameter $c$.
Once problem \eqref{pro:qlm} is solved, we are able to solve problem \eqref{pro:min-max} and the primal problem \eqref{pro:mv} by considering
\begin{align}
\label{pro:theta}
\max_{\theta \in \Rb} \; J_3(x_0, \lam_0; \theta) :=
\max_{\theta \in \Rb} \left\{ \Jc_2(x_0, \lam_0; \xi - \theta) - \theta^2 \right\} ,
\end{align}
where $\xi$ is the target expected terminal wealth in problem \eqref{pro:mv}. To be precise, given the solution $\theta^*$ of \eqref{pro:theta}, $\pi^* : = \pi^*_c|_{c = \xi - \theta^*}$ is the efficient strategy to problem \eqref{pro:mv}, and $({\mathbb E} [(X^*(T)-\xi)^2], \xi)$ is the corresponding efficient frontier,
where $X^*(T)$ is the terminal wealth associated with $\pi^*$ and $\xi \geq x_0 e^{rT}$.

\section{Heuristic derivations}
\label{sec:deri}

As explained above, the key to solving the (constrained) mean-variance problem \eqref{pro:mv} is to analyze the (unconstrained) quadratic-loss minimization problem \eqref{pro:qlm}.
This section is devoted to the investigations of problem \eqref{pro:qlm}.
We obtain a candidate solution of $\pi^*_c$ to problem \eqref{pro:qlm} in \eqref{eq:pi_new}.

Notice that the objective functional $J_2$ of problem \eqref{pro:qlm} is the expected value of a quadratic function and no longer has the ``annoying'' equality constraint, which is different from the objective $J$ of the original MV problem \eqref{pro:mv}.
Namely, problem \eqref{pro:qlm} is a standard stochastic control problem.
As such, several approaches are available to solve such a problem, including the well-known Hamilton-Jacobi-Bellman (HJB) method (via the application of the dynamic programming principle).
Here we apply the stochastic maximum principle (see, e.g., \cite{yong1999stochastic}[Chapter 3] for standard theory) to tackle problem \eqref{pro:qlm}. One may refer to Appendix C for the connection between the stochastic maximum principle approach and the HJB equation approach.
One of the advantages of using the stochastic maximum principle is that it is naturally related to BDSE that is pivotal to discuss
the solvability of \eqref{eq:PDE-g} and design the algorithm for the numerical analysis in Section \ref{sec:exm}.

Since the standard maximum principle only serves as a necessary condition for optimality, the analysis in this section is heuristic and $\pi^*_c$ in \eqref{eq:pi_new} is a \emph{candidate} for the optimal solution to problem \eqref{pro:qlm}.
We do not attempt to verify that $\pi^*_c$ given by \eqref{eq:pi_new} solves problem \eqref{pro:qlm}, since our ultimate goal is to solve the MV problem \eqref{pro:mv}.
(Under additional assumptions, the stochastic maximum principle becomes sufficient and indeed leads to an optimal strategy; see for instance \cite{oksendal2005applied}[Theorem 5.4] and \cite{yong1999stochastic}[Chapter 3, Section 5].)
Instead, the main results of this section will pave the way to Theorem \ref{thm:op} in the next section, which eventually provides a complete solution to the MV problem \eqref{pro:mv}.
The candidate solution $\pi^*_c$ in \eqref{eq:pi_new} to problem \eqref{pro:qlm} has a similar representation form as the efficient strategy $\pi^*$ in \eqref{eq:efficient-strategy} to problem \eqref{pro:mv}.
The relation between $\pi^*_c$ and $\pi^*$ will be established in Section \ref{sec:op}
by using the Lagrangian technique as in \cite{lim2002mean} and \cite{zhou2003markowitz}.

In what follows, we carry out heuristic derivations in \emph{four} steps to solve the quadratic-loss minimization problem \eqref{pro:qlm}.
The main arguments are similar to those in \cite{shen2014optimal}[Lemma 3.1], where the authors apply the stochastic maximum principle to solve a mean-variance investment and reinsurance problem with delay.
Please refer to \cite{oksendal2005applied}[Section 5.2.5] for the application of the stochastic maximum principle to a standard MV portfolio selection problem under the Black-Scholes model \emph{without} contagion risk.
Note that the results in \cite{oksendal2005applied}[Section 5.2.5] are also heuristic or at least informal, as they do not verify the optimality and admissibility of the obtained candidate strategy.
Finally, we remark that, since the Hawkes intensity process $\lam$ is unbounded, the existing results of the stochastic maximum principle for jump-diffusion control problems cannot be directly applied to verify the optimality, which will be deferred to the next section by using the completing the square technique in LQ control theory.

\vspace{2ex}
\noindent
\textbf{Step 1.} Write down the associated Hamiltonian system.

We define a new process $\wx:=\{\wx(t)\}_{t \in [0,T]}$ by
\begin{align}
\label{eq:X_hat}
\wx(t) : = X (t) - c e^{- r (T - t)},
\end{align}
where $c \in \Rb$ is the same free parameter in problem \eqref{pro:qlm}.
Notice that $\wx$ has the same flavor as the auxiliary process $y$ in \cite{zhou2000continuous}.
It is clear that $\wx$ is defined for any $\pi \in \Ac$, and here we suppress such dependency for
notational brevity. By applying It\^o's formula to $\wx$, we obtain
\begin{align}
\label{eq:dX_hat}
\dd {\widehat X} (t) = \big ( r {\widehat X} (t) + \pi (t)^\top B \big ) \dd t + \pi (t)^\top \sigma \, \dd W (t)
+ \int_{(-1, \infty)^m} \pi (t)^\top \eta (z) {\widetilde N} (\dd t, \dd z),
\end{align}
with the initial value ${\widehat X} (0) = x_0 - c e^{- r T}$. The objective functional $J_2$ of problem \eqref{pro:qlm} can be rewritten as
\begin{align}
J_2 (x_0, \lambda_0; \pi, c )
= {\mathbb E} \big [ (X (T) - c)^2  \big ] = \Eb \big[ \wx^2(T) \big] .
\end{align}

We treat process $\wx$ as the state process, and its dynamics equation \eqref{eq:dX_hat} suggests that the Hamiltonian of problem \eqref{pro:qlm} reads as
\begin{align}
\label{eq:Hami}
{\cal H} (t, x, \pi, p, q, u) := \big ( r x + \pi^\top B \big ) p + \pi^\top \sigma q
+ \int_{(-1, \infty)^m} \pi^\top \eta (z) \, \mathrm{Diag} [\lambda (t) \bullet \nu (\dd z)] \, u (t, z),
\end{align}
where operator $\bullet$ is defined by \eqref{eq:bullet}.
The triplet $(p,q,u)$ is the so-called \emph{adjoint process}.
Notice that $p:=\{p(t)\}_{t \in [0,T]} \in \Rb$, $q:=\{q(t)\}_{t \in [0,T]} \in \Rb^n$, and $u:= \{u(t,z)\}_{(t, z) \in [0, T] \times (- 1, \infty)^m} \in \Rb^m$.

\vspace{2ex}
\noindent
\textbf{Step 2.} Find the adjoint process $(p,q,u)$.

According to the standard control theory (see, e.g., \cite{yong1999stochastic} and \cite{oksendal2005applied}), the adjoint process $(p,q,u)$ satisfies the following adjoint equation:
\begin{align}\label{eq:adjoint}
	\left\{
	\begin{aligned}
		\dd p (t) &= - r p (t) \, \dd t + q (t)^\top \, \dd W (t) + \int_{(-1, \infty)^m} u (t, z)^\top {\widetilde N} (\dd t, \dd z) , \\
		p (T) &= 2 {\widehat X} (T) = 2 ( X (T) - c )  ,
	\end{aligned}
	\right.
\end{align}
where $r$ is the risk-free interest rate, $W$ is the $n$-dimensional Brownian motion, and $\widetilde{N}$ is the compensated random measure defined in \eqref{eq:wt_N}. In fact, the drift term in \eqref{eq:adjoint} is found by
\begin{align}\label{eq:adjoint-drift}
- {\cal H}_x (t, {\widehat X} (t), \pi (t), p (t), q (t), u (t)) = - r p (t) .
\end{align}

To solve the adjoint equation \eqref{eq:adjoint}, we try an {\it ansatz} for $p$ in the form of
\begin{align}\label{eq:p-ansatz}
p (t) = Y (t) {\widehat X} (t) ,
\end{align}
where $Y$ is the first component of the solution $(Y, V)$ to the following
backward stochastic differential equation (BSDE):
\begin{align}\label{eq:BSDE-(Y,V)}
\left\{
\begin{aligned}
\dd Y (t) &= - f (t, Y (t), V (t)) \, \dd t + V (t)^\top \, \dd {\widetilde N} (t) , \\
Y (T) &= 2 .
\end{aligned}
\right.
\end{align}
In \eqref{eq:BSDE-(Y,V)}, the process $V$ is not a priori known. Instead, $Y$ and $V$ need to be solved at the same time. The driver $f$ of the BSDE \eqref{eq:BSDE-(Y,V)} is unknown at the moment, but will be determined later. Notice that the BSDE \eqref{eq:BSDE-(Y,V)} is driven by the jump process $\wn$ only.

Using \eqref{eq:dX_hat} and \eqref{eq:BSDE-(Y,V)}, and the stochastic product rule on $p = Y \wx$, we get
\begin{align}
\dd p (t) &= {\widehat X} (t-)\, \dd Y (t) + Y (t-) \, \dd {\widehat X} (t) + \dd [Y (t), {\widehat X} (t)] \\
&= \bigg [ - f (t, Y (t), V (t)) {\widehat X} (t) + Y (t) ( r {\widehat X} (t) + \pi (t)^\top B )  \\
&\quad + \sum^k_{i = 1} \sum^m_{l = 1} \int_{(- 1, \infty)} \pi_i (t) \eta_{il} (z_l) V_{l} (t) \lam_l(t) \nu_l (\dd z_l) \bigg ] \dd t + Y (t) \pi (t)^\top \sigma \, \dd W (t) \label{eq:p-ansatz-BSDE} \\
&\quad + \sum^m_{l = 1} \int_{(- 1, \infty)} \bigg [ {\widehat X} (t-) V_{l} (t)
+ \sum^k_{i = 1}  \pi_{i} (t) \eta_{il} (z_l) \, \big ( Y (t-)
+ V_{l} (t) \big ) \bigg ] {\widetilde N}_l (\dd t, \dd z_l) .
\end{align}
By comparing the dynamics of $p$ in \eqref{eq:adjoint} and \eqref{eq:p-ansatz-BSDE}, we obtain
\begin{align}
- r Y (t) {\widehat X} (t) &= - f (t, Y (t), V (t)) {\widehat X} (t)
+ Y (t) ( r {\widehat X} (t) + \pi (t)^\top B ) \\
&\quad + \sum^k_{i = 1} \sum^m_{l = 1} \int_{(- 1, \infty)} \pi_{i} (t) \eta_{il} (z_l) V_{l} (t) \lam_l (t) \nu_l (\dd z_l),  \label{eq:find-driver} \\
q (t) &= \sigma^\top \pi (t) Y (t) ,  \label{eq:q-ansatz} \\
 u_{l} (t, z_l) &= {\widehat X} (t-) V_{l} (t)
+ \sum^k_{i = 1}  \pi_{i} (t) \eta_{il} (z_l) \, \big ( Y (t-)
+ V_{l} (t) \big ), \qquad \forall \, l \in \Mb. \label{eq:u-ansatz}
\end{align}

In the above system of equations \eqref{eq:find-driver}, \eqref{eq:q-ansatz} and \eqref{eq:u-ansatz}, the unknowns are $Y$, $V=(V_1,\ldots, V_m)^\top$, $f$, $q$, and $u=(u_1,\cdots,u_m)^\top$.
The three equations are not enough to fully solve all the five unknowns.
At this stage, the adjoint equation \eqref{eq:adjoint} is partially solved subject to $Y$ and $V$ satisfying the BSDE \eqref{eq:BSDE-(Y,V)}.
\\[2ex]
\textbf{Step 3.} Find a candidate for the optimal solution to problem \eqref{pro:min-max}.

Under the maximum principle, the candidate for the optimal solution is given by
\begin{align}
	\pi_c^*(t) = \argmax_{\pi \in \Rb^k} \; \Hc \left(t, {\widehat X}^*(t), \pi, p^*(t), q^*(t), u^*(t,\cdot) \right),
\end{align}
where $\Hc$ is defined in \eqref{eq:Hami} and ${\widehat X}^*$ and $(p^*, q^*, u^*)$ are the wealth process
and the adjoint process associated with strategy $\pi^*_c$. The subscript in $\pi^*_c$ indicates that problem \eqref{pro:qlm} is parameterized by $c \in \Rb$, and note that the same parameter $c$ also enters the definition of $\wx$ in \eqref{eq:X_hat}.
As a result,
we obtain  $\pi_c^*$ via the first-order condition
\begin{align}\label{eq:Hamiltonian-FOC}
p(t)  B + \sigma q (t) + \int_{(-1, \infty)^m} \eta (z) \, \mathrm{Diag} [ \lam (t) \bullet \nu (\dd z) ] \, u (t, z) = 0,
\end{align}
where $B$ is the risk premium vector of dimension $k \times 1$, $\sig$ is the volatility matrix of dimension $k \times n$, and $\eta$ is the jump size matrix of dimension $k \times m$ (see their definitions in \eqref{eq:notation_vec} and \eqref{eq:eta}).

Substituting \eqref{eq:p-ansatz}, \eqref{eq:q-ansatz} and \eqref{eq:u-ansatz} into the above equation  \eqref{eq:Hamiltonian-FOC}, we get
\begin{align}
\label{eq:pi_can}
\pi_c^* (t) = - \Theta (t, Y (t-), V (t))^{-1} \Zc(t, Y(t-), V(t)) \, {\widehat X}^* (t-), \qquad
\end{align}
where the $k \times k$ matrix $\Theta$ and the $k$-dimensional vector $\Zc$ are defined by
\begin{align}
\Theta (t, Y (t-), V (t)) &:= Y (t-) \bigg ( \sigma \sigma^\top
+ \int_{(-1, \infty)^m} \eta (z)\,  \mathrm{Diag} [ \lam (t) \bullet \nu (\dd z) ] \eta (z)^\top \bigg ) \\
&\qquad + \int_{(-1, \infty)^m} \eta (z) \, \mathrm{Diag} [ V (t) \bullet \lam (t) \bullet \nu (\dd z) ] \, \eta (z)^\top, \label{eq:Theta} \\
\text{ and} \quad \Zc(t, Y(t-), V(t)) &:= B Y (t-) + \int_{(-1,\infty)^m} \eta (z) \, \mathrm{Diag} [ \lam (t) \bullet \nu (\dd z) ] \, V(t). \label{eq:Zc}
\end{align}

We end this step by noting that $\pi_c^*$ given by \eqref{eq:pi_can} is in closed form, once $Y$ and $V$ are identified, which will be achieved in the next step.
\\[2ex]
\textbf{Step 4.} Solve the BSDE \eqref{eq:BSDE-(Y,V)} to obtain $Y$ and $V$.

We first plug $\pi_c^*$, given by \eqref{eq:pi_can}, back to \eqref{eq:find-driver}, and identify the driver $f$ of the BSDE \eqref{eq:BSDE-(Y,V)} as follows:
\begin{align}\label{eq:driver}
f (t, Y (t), V (t))
& = \left( 2 r - B^\top \Theta^{-1} \Zc  \right) Y (t)
 - \int_{(-1,\infty)^m} V (t)^\top \, \mathrm{Diag}[ \lam (t) \bullet \nu (\dd z) ] \eta (z)^\top \Theta^{-1} \Zc \nonumber \\
& = 2 r Y (t) - \Zc (t, Y (t), V (t))^\top \Theta (t, Y (t), V (t))^{-1} \Zc (t, Y (t), V (t)) ,
\end{align}
where we have suppressed the arguments of both $\Theta$ and $\Zc$ in the first equality above.

Now with the driver $f$ given by \eqref{eq:driver}, the BSDE \eqref{eq:BSDE-(Y,V)} is fully specified.
We try the following ansatz  for the solution $(Y,V)$ to \eqref{eq:BSDE-(Y,V)}:
\begin{align}
Y (t) &= 2 e^{2 r (T - t) + g (t, \, \lambda (t)) }, \label{eq:Y-ansatz} \\
V_{l} (t) &= 2 e^{2 r (T - t) + g (t, \, \lambda (t))} \big ( e^{\, g (t, \, \lambda (t)+ \beta_{(l)}) - g (t, \, \lambda (t))} - 1 \big ) ,  \qquad l \in \Mb, \label{eq:V-ansatz}
\end{align}
where $g (t, \lambda)$ is yet to be determined and $\beta_{(l)}$ is the $l^{th}$ column of the matrix $\beta = [\beta_{jl}]_{j, l \in \Mb}$, i.e.,
\begin{align}
\beta_{(l)} : = (\beta_{1l}, \beta_{2l}, \ldots, \beta_{ml})^\top.
\end{align}
Recall that $\beta$ captures the impact of mutual-excitation on the intensity process $\lam$; see \eqref{eq:intensity} and \eqref{eq:vec}.

Applying It\^o's formula to $Y$ in \eqref{eq:Y-ansatz}, we have
\begin{align}\label{eq:BSDE-(Y,V)-ansatz}
\dd Y (t) = Y (t) \bigg [ &- 2 r + g_t + g_\lambda^\top \big (  \alpha \lambda_{\infty} + (\beta - \alpha) \lambda (t) \big )
+ \sum^m_{l = 1} \bigg ( e^{g (t, \, \lambda (t)+\beta_{(l)}) - g (t, \, \lambda (t))} - 1\\
& - \sum^m_{j = 1}  \beta_{jl} g_{\lambda_l} \bigg ) \lambda_l (t) \bigg ]  \, \dd t
 + Y (t-) \sum^m_{l = 1} \bigg ( e^{g (t, \, \lambda (t)+\beta_{(l)}) - g (t, \, \lambda (t))} - 1 \bigg ) \dd {\widetilde N}_{l} (t),
\end{align}
where $g_\cdot$ denotes the partial derivative of $g$ with respect to the corresponding argument (e.g., $g_t = \frac{\partial g}{\partial t}(t, \lam) $), and the arguments of $g$ are suppressed in its partial derivatives.

To simplify notations, let us introduce
\begin{align}
U (t, \lambda ) :=& \ \bigg ( e^{g (t, \, \lambda +\beta_{(1)}) - g (t, \, \lambda )} - 1,
\ldots, e^{g (t, \, \lambda +\beta_{(m)}) - g (t, \, \lambda )} - 1 \bigg )^\top, \label{eq:U(b)} \\
 \Gamma (t, \lambda ) :=& \
\frac{1}{Y(t-)} \, \Theta(t, Y (t-), V (t)) \\
=& \ \sig \sig^\top + \int_{(-1, \infty)^m} \; \eta(z) \, \mathrm{Diag}[(U(t, \lambda) + \bm{1}_m) \bullet \lam  \bullet \nu(\dd z)] \, \eta(z)^\top,
\label{eq:Gamma(t,b)} \\
\wz(t, \lambda) :=& \ \frac{1}{Y(t-)} \, \Zc(t, Y(t-), V(t))  \\
=& \ B + \int_{(-1, \infty)^m} \; \eta(z) \, \mathrm{Diag}[\lam  \bullet \nu(\dd z)] \, U (t, \lambda ),
\label{eq:Z_hat}
\end{align}
where  $\Theta$ and $\Zc$ are defined respectively by \eqref{eq:Theta} and \eqref{eq:Zc}, and $\bm{1}_m$ is the $m$-dimensional vector with all elements equal to 1.
It is important to notice that, with the  ansatz
of $Y$ and $V$ in \eqref{eq:Y-ansatz} and \eqref{eq:V-ansatz}, $\Gamma$ and $\wz$ defined above are indeed functions of $t$ and $\lambda$ only, and are independent of $Y$ and $V$.
By the above definitions, we easily see
\begin{align}
V (t) = Y (t) \, U (t, \lambda (t)), \quad
\Theta (t, Y (t-), V (t)) = Y (t-) \, \Gamma (t, \lambda (t)), \quad
\Zc(t, Y(t-), V(t)) = Y(t-) \wz(t, \lambda (t)).
\end{align}
In addition, we reorganize $\pi_c^*$ in \eqref{eq:pi_can} and obtain
\begin{align}
\label{eq:pi_new}
\pi_c^*(t) = - \Gamma(t, \lambda (t))^{-1} \, \wz(t, \lambda (t)) \, \wx^*(t-).
\end{align}

Using $U$ and $\Gamma$ in \eqref{eq:U(b)} and \eqref{eq:Gamma(t,b)}, we rewrite the BSDE \eqref{eq:BSDE-(Y,V)} in the differential form
\begin{align}
\dd Y (t) = &- Y (t) \bigg [ 2 r  - \wz(t, \lambda(t))^\top \Gamma (t, \lambda (t))^{-1} \wz(t, \lambda (t)) \bigg ] \dd t
+ Y (t-) U (t, \lambda (t))^\top \, \dd {\widetilde N} (t) . \label{eq:BSDE-(Y,V)-comp}
\end{align}

By matching the drift terms in \eqref{eq:BSDE-(Y,V)-ansatz} and \eqref{eq:BSDE-(Y,V)-comp}, we obtain
the following non-local partial differential equation of $g: [0,T] \times \Rb_+^m \to \Rb$
\begin{align}\label{eq:PDE-g}
\begin{cases}
g_t(t, \lam) + g_\lambda(t,\lam)^\top \, \alpha (\lambda_{\infty} - \lambda )
	+ U (t, \lambda)^\top \lambda = \wz(t, \lambda)^\top \Gamma (t, \lambda)^{-1} \wz(t, \lambda), & t < T, \\
g (t, \cdot) = 0, & t=T .	
\end{cases}
\end{align}
Recall matrix $\alpha$ and vector $\lam_\infty$ are defined by \eqref{eq:vec}, while $U$, $\Gamma$ and $\wz$ are defined respectively by \eqref{eq:U(b)}, \eqref{eq:Gamma(t,b)} and \eqref{eq:Z_hat}.
In \eqref{eq:PDE-g}, $U$ has a non-local dependence on $g$ (see \eqref{eq:U(b)}), which originates 
from the mutual-excitation feature of the Hawkes process. This completes the 4-step procedure.

\vspace{2ex}

We next present a technical result regrading the bounds of $e^{g(t, \lam)}$, where $g$ solves \eqref{eq:PDE-g}.

\begin{lemma}
\label{lem:g}
Let $g (\cdot, \cdot)$ be the solution to \eqref{eq:PDE-g}. Then there exists $\epsilon > 0$ such that
\begin{align}
	\label{eq:g-bound}
	\epsilon \leq e^{g (t, \lam (t))} \leq 1 , \qquad \forall \, t \in [0, T] , \qquad \text{and} \qquad
	\epsilon \leq e^{g (0, \lam_0)} < 1 .
\end{align}
where $\lam(\cdot)$ is the intensity process satisfying the Hawkes model \eqref{eq:intensity}.
\end{lemma}
\begin{proof}
	See Appendix \ref{app:lem}.
\end{proof}

We provide the solvability results of \eqref{eq:PDE-g} in the following lemma.

\begin{lemma}
\label{lem:g-solvability}
The non-local partial differential equation \eqref{eq:PDE-g} admits a unique solution.
\end{lemma}

\begin{proof}
	See Appendix \ref{app:pde}.
\end{proof}

We end this section with several important remarks on the above heuristic analysis of problem \eqref{pro:qlm} and explain how the analysis here will help solve the MV problem \eqref{pro:mv} in Section \ref{sec:op}.
\\[-5ex]
\begin{itemize}
	\item Notice that problem \eqref{pro:qlm} aims to minimize the expectation of a convex function of $\wx(T)$ and the admissible set $\Ac$ is also convex, as easily seen from Definition \ref{def:adm}. In consequence, we can also apply the HJB equation approach to solve problem \eqref{pro:qlm}; see Appendix \ref{app:hjb}. Indeed, the solution obtained by the HJB equation approach matches that obtained by the stochastic maximum principle approach,
which is due to the inherent relationship between the two approaches. One can refer to \cite{oksendal2005applied} for such relationship under the Poisson jump case.
	If the jump process is not a Hawkes process but a standard Poisson process, one can follow the standard procedure to verify that
	$\pi^*_c$ given by \eqref{eq:pi_new} is indeed the optimal solution to problem \eqref{pro:qlm}. We refer interested readers to \cite{yong1999stochastic}[Chapter 3, Section 5] for the sufficient conditions of optimality on general stochastic control problems without jumps. However, jumps in our model are derived from a multivariate Hawkes process and, due to the unboundedness of its intensity process $\lam$, existing results cannot be readily applied to confirm the optimality of $\pi^*_c$.
	
	 \item Given that the wealth SDE \eqref{eq:wealth} is linear and the objective $J_2$ is quadratic, problem \eqref{pro:qlm} is a linear-quadratic (LQ) control problem. By LQ theory, we naturally conjecture that $Y \wx^2$ is a submartingale for any $\pi \in \Ac$ and a martingale under $\pi^*_c$.
	 	Assume for now that such a conjecture is correct, which will be verified in the next section, then we immediately have
	 	\begin{align}
	 		\frac{1}{2} \Eb \Big[ Y(T)  \big( \wx^*(T) \big)^2 \Big] = 	\frac{1}{2}  Y(0)  \big( \wx^*(0) \big)^2  \le \frac{1}{2} \Eb \Big[ Y(T)  \big( \wx(T) \big)^2 \Big],
	 	\end{align}
 	where $\wx^*$ denotes the wealth process under the ``optimal'' strategy $\pi^*_c$. As a result, we observe after noticing $Y(T) = 2$ that
	 	\begin{align}
	 		\label{eq:V2}
	 	\Jc_2(x_0, \lam_0; c) = \min_{\pi \in \Ac} \, \frac{1}{2} \Eb \Big[ Y(T)  \big( \wx(T) \big)^2 \Big]
	 	= e^{\, 2rT + g(0, \lam_0)} \big(x_0 - c \, e^{-rT} \big)^2 ,
	 	\end{align}
 where we have used \eqref{eq:X_hat} and \eqref{eq:Y-ansatz}.

    \item Once the value function to problem \eqref{pro:qlm} is found via \eqref{eq:V2}, we next solve problem \eqref{pro:theta}, a simple quadratic problem of  $\theta$ (i.e., the Lagrangian multiplier), and obtain its solution $\theta^*$. Now applying the arguments at the end of Section \ref{sec:prob} (see also \cite{lim2002mean}
        and \cite{zhou2003markowitz}), we conclude that $\pi^*_c$ with $c$ replaced by $\xi - \theta^*$ is the efficient strategy to the MV problem \eqref{pro:mv}, upon a successful verification of $\pi^*_{\xi - \theta^*} \in \Ac$.

\end{itemize}

\section{Efficient strategy and efficient frontier}
\label{sec:op}

In this section, we provide a complete solution to the mean-variance portfolio selection problem \eqref{pro:mv}. Theorem \ref{thm:op} obtains the efficient strategy and the efficient frontier to problem \eqref{pro:mv}, both in semi-closed form. To prove this theorem, we follow the key steps outlined in the above remarks in the previous section.
In particular, to show $Y \wx^2$ is a martingale under the optimal strategy $\pi^*_c$, we apply the idea of ``completing the square'' in LQ control theory (see \cite{yong1999stochastic} and \cite{zhou2000continuous}). A technical difficulty in the proof of Theorem \ref{thm:op} is to verify that the candidate efficient strategy is admissible.

\begin{theorem}
\label{thm:op}
Let $g (\cdot, \cdot)$ be the solution to  \eqref{eq:PDE-g}, and
$\Gamma (\cdot, \cdot)$ and $\wz (\cdot, \cdot)$ be defined by \eqref{eq:Gamma(t,b)} and \eqref{eq:Z_hat}, respectively.
The efficient strategy $\pi^*=\{\pi^*(t)\}_{t \in [0,T]}$ of the mean-variance portfolio selection problem \eqref{pro:mv} is
given by
\begin{align}\label{eq:efficient-strategy}
\pi^* (t) =
- \Gamma (t, \lambda (t))^{-1} \, \wz(t, \lambda(t)) \, \left( X^*(t-) - (\xi - \theta^*) e^{- r (T - t)} \right),
\end{align}
and the efficient frontier is given by
\begin{align}\label{eq:efficient-frontier}
\mathrm{Var} [X^* (T)] = \dfrac{e^{g (0, \lambda_0)}}{1 - e^{g (0, \lambda_0)}} \left(x_0 e^{r T} - \xi \right)^2,
\qquad \xi \geq x_0  e^{rT} ,
\end{align}
where the corresponding Lagrange multiplier $\theta^*$ is obtained by
\begin{align}
\label{eq:theta_op}
\theta^* = \dfrac{e^{g (0, \lambda_0)}}{1 - e^{g (0, \lambda_0)}} \left(x_0 e^{r T} - \xi \right) .
\end{align}
Here, $X^*$ is the solution to the SDE \eqref{eq:wealth} under strategy $\pi^*$, and is called the optimal wealth process.
\end{theorem}

\begin{proof}
We divide the proof into three steps.	

\vspace{1ex}	
\noindent
\textbf{Step 1:} Apply the ``completing the square'' to obtain $\pi^*_c$.

Recall that the dynamics of $\wx$ is given by \eqref{eq:dX_hat} and that of $Y$ in \eqref{eq:BSDE-(Y,V)-comp} from Section \ref{sec:deri}.
By applying It\^o's formula to $\wx^2$ and $Y \wx^2$, we obtain
\begin{align}
\dd {\widehat X}^2 (t) &= \bigg ( 2 r {\widehat X}^2 (t) + 2 {\widehat X} (t) \pi (t)^\top B
+ \pi (t)^\top \sigma \sigma^\top \pi (t)   + \int_{(-1, \infty)^m} \pi (t)^\top \eta (z) \, \mathrm{Diag} [\lam (t) \bullet \nu (\dd z)] \, \eta (z)^\top \pi (t) \bigg ) \dd t  \\
&\quad+ 2 {\widehat X} (t) \pi (t)^\top \sigma \, \dd W (t) + \int_{(-1, \infty)^m} 2 {\widehat X} (t-) \pi (t)^\top \eta (z) {\widetilde N} (\dd t, \dd z) \\
&\quad + \int_{(-1, \infty)^m} \pi (t)^\top \eta (t, z) \, \mathrm{Diag} [{\widetilde N} (\dd t, \dd z)] \, \eta (z)^\top \pi (t) ,
\end{align}
and
\begin{align}
\dd Y (t) {\widehat X}^2 (t)
&= Y (t) \bigg [ \pi (t) + \Gamma (t, \lam (t))^{-1} \wz(t, \lam (t)) {\widehat X} (t) \bigg ]^\top
 \Gamma (t, \lam (t)) \underbrace{\bigg [ \pi (t) + \Gamma (t, \lam (t))^{-1} \wz(t, \lam (t)) {\widehat X} (t) \bigg ]}_{:= \Uc(t,\lam(t))}
 \dd t  \\
&\quad + 2 Y (t) {\widehat X} (t) \pi (t)^\top \sigma \, \dd W (t) + \int_{(-1, \infty)^m} Y (t-) \pi (t)^\top \eta (z) \, \mathrm{Diag} [U (t, \lam (t)) \bullet {\widetilde N} (\dd t, \dd z)] \,
\eta (z)^\top \pi (t)\\
&\quad + \int_{(-1, \infty)^m} Y (t-) {\widehat X} (t-) \big ( 2 \pi (t)^\top \eta (z)
+ {\widehat X} (t-) U (t, \lam (t)) \big ) {\widetilde N} (\dd t, \dd z) . \label{eq:SDE-YX2}
\end{align}

Define a sequence of stopping times $\{\tau_i\}_{i=1,2,\ldots}$ by
\begin{align}
\tau_i := \inf \left\{ t \ge 0 \, \Big| \, |Y(t) \wx(t) |>i \right\}, \qquad i=1,2,\ldots \, .
\end{align}
It is easily seen that $\tau_i \uparrow \infty$ and $(T \wedge \tau_i) \uparrow T$, as $i \to \infty$. We apply the localization technique, i.e., integrating from $0$ to $T \wedge \tau_i$
and taking expectations on both sides of \eqref{eq:SDE-YX2}, and eventually get
\begin{align}\label{eq:expectation-localization}
\frac{1}{2} {\mathbb E} \big [ Y (T \wedge \tau_i) (X (T \wedge \tau_i) - c)^2 \big ] - \frac{1}{2} Y (0) \left(x_0 - c e^{- r T} \right)^2 &=
\frac{1}{2} {\mathbb E} \bigg [ \int^{T \wedge \tau_i}_0 Y (t) \, \Uc(t, \lam (t))^\top
 \Gamma (t, \lam (t)) \, \Uc(t, \lam(t)) \dd t \bigg ],
\end{align}
where we have used the integrability conditions in Definition \ref{def:adm}.
Note that $X (\cdot) \in {\cal S}^2_{\cal F} (0, T; {\mathbb R})$, $Y (\cdot)$ is bounded, and the integrand on
the right-hand of the above equality is positive. Then by the dominated and monotone convergence theorems
to the left-hand side and the right-hand side of \eqref{eq:expectation-localization}, we obtain, when sending $i$ to $\infty$, that
\begin{align}
{\mathbb E} \big [ (X (T) - c)^2 \big ] - \frac{1}{2} Y (0) \left(x_0 - c e^{- r T} \right)^2
= \frac{1}{2} {\mathbb E} \bigg [ \int^T_0  Y (t) \, \Uc(t, \lam (t))^\top \,
 \Gamma (t, \lam (t)) \, \Uc(t, \lam(t)) \dd t  \bigg ] ,
\end{align}
where we have used the terminal value $Y(T)=2$.

Therefore, by setting $\Uc(t, \lam(t)) = 0$, we find the (candidate) optimal control by
\begin{align}\label{eq:optimal-control}
\pi_c^* (t) = - \Gamma (t, \lam (t))^{-1} \underbrace{\bigg ( B
+ \int_{(-1, \infty)^m} \eta (z)\, \mathrm{Diag}[\lam (t) \bullet \nu (\dd z)] \, U (t, \lam (t)) \bigg )}_{=\wz(t, \lam (t))} \underbrace{\left( X^*(t-) - c e^{- r (T - t)} \right)}_{ = \wx^*(t-)}, \quad
\end{align}
which coincides with the one in \eqref{eq:pi_new}.

We end this step with an important observation. In Section \ref{sec:deri}, we obtain $\pi^*_c$ in \eqref{eq:pi_new} by the stochastic maximum principle as a necessary condition for optimality.
	Here by the so-called ``completing the square'' technique, we easily see that $\pi^*_c$ given by \eqref{eq:optimal-control} is indeed optimal to the quadratic-loss minimization problem \eqref{pro:qlm} (a sufficient condition for optimality).
	Note that sufficiency is guaranteed by the facts that $Y >0$ and $\Gamma$ is positive definite; see \eqref{eq:Y-ansatz} and \eqref{eq:Gamma(t,b)}.

\vspace{1ex}
\noindent
\textbf{Step 2:} Solve problem \eqref{pro:theta} to get $\theta^*$ and obtain the efficient strategy and
the efficient frontier.

Recall that problem \eqref{pro:theta} is to maximize $\Jc_2(x_0, \lam_0; \xi - \theta) - \theta^2$ over all $\theta \in \Rb$, where $\Jc_2$ is  the value function to problem \eqref{pro:qlm} and is obtained by (see \eqref{eq:V2})
\begin{align}\label{eq:optimal-cost}
\Jc_2(x_0, \lam_0;c)
= {\mathbb E} \big [ ({\widehat X}^* (T))^2 \big ] = {\mathbb E} \big [ (X^* (T) - c)^2 \big ] = \frac{1}{2} Y (0) \left(x_0 - c e^{- r T} \right)^2 .
\end{align}
Now noting $c = \xi - \theta$, we obtain $\theta^*$ by
\begin{align}
\theta^* 
&= \argmax_{\theta \in \Rb} \; \bigg \{ \frac{1}{2} Y (0) \left(x_0 - (\xi - \theta) e^{- r T} \right)^2 - \theta^2 \bigg \} \\
&= \dfrac{ e^{- 2 r T} Y (0)}{2 -  e^{- 2 r T} Y (0)} \left(x_0 e^{rT} - \xi \right),
\end{align}
which reduces to \eqref{eq:theta_op} by using \eqref{eq:Y-ansatz}.
The result of Lemma \ref{lem:g} guarantees that the first-order condition to problem \eqref{pro:theta} is also sufficient and the above $\theta^*$ is indeed the optimal solution to problem \eqref{pro:theta}.

With $\theta^*$ obtained as in \eqref{eq:theta_op}, we immediately get the efficient strategy $\pi^*$ by using  $\pi^* = \pi^*_{\xi - \theta^*}$ and the general expression of $\pi^*_c$ in \eqref{eq:optimal-control}.
Using the definition of $\wx$ in \eqref{eq:X_hat} again then leads to the result of $\pi^*$ in \eqref{eq:efficient-strategy}.
Finally, by the Lagrangian duality theorem, we obtain the efficient frontier in \eqref{eq:efficient-frontier} via
\begin{align}
	\mathrm{Var}[X^*(T)] = {\mathbb E} \big [ \big(X^* (T) - (\xi - \theta^*) \big)^2 \big ] - (\theta^*)^2.
\end{align}

\vspace{1ex}
\noindent
\textbf{Step 3:} Verify that $\pi^*$ given by \eqref{eq:efficient-strategy} is admissible (i.e., $\pi^* \in \Ac$).

The previous two steps have justified all the results in Theorem \ref{thm:op}, only contingent on $\pi^* \in \Ac$. Hence, in the final step, we verify that the efficient strategy $\pi^*$ in \eqref{eq:efficient-strategy} is admissible.

 To start, by using \eqref{eq:optimal-cost} and $e^{g (0, \lambda_0)}<1$ from Lemma \ref{lem:g}, we have
\begin{align}
{\mathbb E} \big [ (X^* (T) - (\xi - \theta^*))^2 \big ]
&= \mathrm{Var}[X^*(T)] + (\theta^*)^2  = \dfrac{e^{g (0, \lambda_0)}}{(1 - e^{g (0, \lambda_0)})^2} \left(x_0 e^{r T} - \xi \right)^2 < \infty .
\end{align}
Thus, $\widehat{X}^* (T) = X^* (T) - (\xi - \theta^*)$ is square integrable.
This fact implies that
the adjoint equation \eqref{eq:adjoint}, associated with the efficient strategy $\pi^*$, admits a unique solution such that
$(p^*, q^*, u^*) \in {\cal S}^2_{\cal F} (0, T; {\mathbb R}) \times {\cal L}^2_{\cal F} (0, T; {\mathbb R}^n)
\times {\cal L}^{2,N}_{\cal F} (0, T; {\mathbb R}^m)$.

Next we proceed to show that $Y (t)$ is bounded below from 0.
That is, there exists a positive constant $\delta$ such that $Y (t) \geq \delta$.
Since the efficient strategy does not correspond to a risk-free investment,
$e^{g (0, \lam_0)}$ must be positive. Otherwise, it can be seen from \eqref{eq:efficient-frontier}
that the minimum variance is always zero regardless of the expected return, i.e.,
$\mbox{Var} [X^* (T)] = 0$, for any $\xi \in ( x_0 e^{r T}, + \infty )$, which is unreasonable.
Therefore, using \eqref{eq:e^g}, we obtain
\begin{align}\label{eq:Y-lowerbound}
Y (t) = 2 e^{2 r (T - t) + g (t, \lambda (t)) } \geq 2 e^{2 r T} \cdot e^{g (0, \lambda_0) } : = \delta > 0 .
\end{align}
It then follows from \eqref{eq:p-ansatz} and $p^* \in {\cal S}^2_{\cal F} (0, T; {\mathbb R})$ that
\begin{align}
{\widehat X}^*(t) = \frac{p^*(t)}{Y(t)} \leq \frac{p^*(t)}{\delta}  \quad \text{ for all } t \in [0,T] \quad \Rightarrow \quad \wx^* \in {\cal S}^2_{\cal F} (0, T; {\mathbb R})
\end{align}
and from \eqref{eq:e^g} that $V_l$ is bounded, for all $l \in \mathbb{M}$. Note
\begin{align}
\sum^k_{i = 1}  \pi^*_{i} (t) \eta_{il} (z_l) \, Y (t) = \sum^k_{i = 1}  \pi^*_{i} (t) \eta_{il} (z_l) \, \big ( Y (t-)
+ V_{l} (t) \big ) = u_{l} (t, z_l) - {\widehat X}^* (t-) V_{l} (t) .
\end{align}
Thus, we have
\begin{align}
Y \eta^\top \pi^* &= u  - {\widehat X}^* V \in {\cal L}^{2,N}_{\cal F} (0, T; {\mathbb R}^m) \\
\mbox{and} \quad \eta^\top \pi^* &= \frac{u  - {\widehat X}^* V}{Y} \leq \frac{u  - {\widehat X}^* V}{\delta} \in {\cal L}^{2,N}_{\cal F} (0, T; {\mathbb R}^m) .
\end{align}
Combining the above results, the non-degenerate condition in Lemma \ref{lem:Sigma} and $q^* \in {\cal L}^2_{\cal F} (0, T; {\mathbb R}^n)$ leads to
\begin{align}
\epsilon \, \delta^2 \, {\mathbb E} \bigg [ \int^T_0 |\pi^* (t)|^2 \, \dd t \bigg ]
&\leq {\mathbb E} \bigg [ \int^T_0 |Y (t)|^2 \pi^* (t)^\top \Sigma (t) \pi^* (t) \, \dd t \bigg ] \\
&= {\mathbb E} \bigg [ \int^T_0 |q^* (t)|^2 \, \dd t + \int^T_0 |Y (t) \eta (z)^\top \pi^* (t)|^2_N \, \dd t \bigg ] < \infty .
\end{align}
That is, $\pi^* \in {\cal L}^2_{\cal F} (0, T; {\mathbb R}^k)$.
All the conditions in Definition \ref{def:adm} are now confirmed, and the proof is completed.
\end{proof}

For comparison, let us consider a standard \emph{Poisson-jump-diffusion} market model without contagion risk, in which jumps of asset prices are modeled by an $m$-dimensional Poisson process $N_P$ with a \emph{deterministic} intensity $\lambda_P$.
More specifically, the asset prices follow the SDE:
\begin{align}
\label{eq:dS_vec-Poisson}
\dd S(t) = \mathrm{Diag}[S(t-)] \left( \mu \, \dd t + \sig \, \dd W(t) +  \int_{(-1,\infty)^m} \; \eta(z) \, \wn_P(\dd t, \dd z) \right),
\end{align}
where $\wn_P$ is defined similarly as $\wn$ in \eqref{eq:wt_N}, but with the stochastic intensity $\lambda$ replaced by
the deterministic intensity $\lambda_P$. For a given investment strategy $\pi$, the investor's wealth process evolves according to
\begin{align}\label{eq:wealth-Poisson}
\dd X (t) = \big ( r X (t) + \pi (t)^\top B \big ) \, \dd t + \pi (t)^\top \sigma \, \dd W (t)
+ \int_{(-1, \infty)^m} \pi (t)^\top \eta (z) {\widetilde N}_P (\dd t, \dd z) .
\end{align}
The MV problems with Poisson jumps are considered by \cite{framstad2004mp} and \cite{shen2013mv}.
We directly provide solutions to the MV problem \eqref{pro:mv} in the Poisson-jump-diffusion model \eqref{eq:dS_vec-Poisson} as a corollary, which recovers the results obtained in these papers.

\begin{corollary}
\label{cor:Poi}
In the standard Poisson-jump-diffusion model, described in \eqref{eq:dS_vec-Poisson}, we obtain the efficient strategy to problem \eqref{pro:mv} by
\begin{align}\label{eq:efficient-strategy-Poisson}
\pi^*_P (t) =
- \Gamma_P (\lambda_P (t))^{-1} \, B \, \left( X^*(t-) - (\xi - \theta^*_P) e^{- r (T - t)} \right),
\end{align}
where
\begin{align}
	\label{eq:Ga_P}
\Gamma_P (\lambda) :=
\sig \sig^\top + \int_{(-1, \infty)^m} \; \eta(z) \, \mathrm{Diag}[\lam \bullet \nu(\dd z)] \, \eta(z)^\top,
\end{align}
and
\begin{align}
\theta^*_P = \dfrac{e^{- \int^T_0 B^\top \Gamma_P (\lambda_P (t))^{-1} B \, \dd t}}{1 - e^{- \int^T_0 B^\top \Gamma_P (\lambda_P (t))^{-1} B \, \dd t}} \left(x_0 e^{r T} - \xi \right).
\end{align}
The efficient frontier is given by
\begin{align}
\mathrm{Var} [X^*_P (T)] = \dfrac{e^{- \int^T_0 B^\top \Gamma_P (\lambda_P (t))^{-1} B \, \dd t}}
{1 - e^{- \int^T_0 B^\top \Gamma_P (\lambda_P (t))^{-1} B \, \dd t}} \left(x_0 e^{r T} - \xi \right)^2, \qquad \xi \geq  x_0 e^{r T}.
\end{align}
\end{corollary}

\begin{rmk}
In Corollary \ref{cor:Poi}, the intensity $\lam_P$ of the Poisson process is \emph{deterministic}. In comparison, the intensity $\lam$ of the Hawkes process $N$ is \emph{stochastic}, given by \eqref{eq:dlam_l}. However, if we set $\beta_{lj} = 0$ for all $l, j \in \Mb$, the intensity $\lam$ in \eqref{eq:dlam_l} reduces to the deterministic case,
denoted by $\lam_P$, and the non-local partial differential equation \eqref{eq:PDE-g} becomes an ODE given as follows:
	\begin{align}\label{eq:ODE-g}
		\frac{\dd g(t)}{\dd t}
		= B^\top
		\bigg ( \sig \sig^\top + \int_{(-1, \infty)^m} \; \eta(z) \, \mathrm{Diag}[\lam_P (t) \bullet \nu(\dd z)] \, \eta(z)^\top \bigg )^{-1} B = B^\top \, \Gamma_P(\lam_P(t))^{-1} \, B ,
	\end{align}
where $\Gamma_P$ is defined in \eqref{eq:Ga_P}.
Solving the above ODE yields  the following closed-form solution
	\begin{align}
		\label{eq:g_det}
		g (t) = - \int^T_t \, B^\top
		\Gamma_P(\lam_P(s))^{-1} B \, \dd s .
	\end{align}
By comparing \eqref{eq:efficient-strategy} with \eqref{eq:efficient-strategy-Poisson}, we conclude that the efficient strategy $\pi^*_P$ in the Poisson-jump-diffusion model can be seen as a special case of the efficient strategy $\pi^*$ in the Hawkes-jump-diffusion model. In \eqref{eq:efficient-strategy-Poisson}, the term $\Gamma_P (\lambda_P (t))^{-1} \, B$ represents the product of the precision matrix and
the risk premium vector, which is standard in the literature on portfolio selection problems, while in \eqref{eq:efficient-strategy},
$\Gamma(t, \lambda (t))^{-1} \, \wz(t, \lambda (t))$ is the counterpart term, in which both $\Gamma(t, \lambda (t))^{-1}$ and $\wz(t, \lambda (t))$
are related to the precision matrix $\Sigma (t)^{-1}$ and the risk premium vector $B$, but adjusted by the non-local terms in $U (t, \lambda (t))$.
Please refer to equations \eqref{eq:U(b)}-\eqref{eq:Z_hat} for the definitions of $U (\cdot, \cdot)$, $\Gamma (\cdot, \cdot)$, and $\wz (\cdot, \cdot)$.
\end{rmk}

\section{Numerical analysis}
\label{sec:exm}

In this section, we conduct numerical analysis to illustrate the theoretical findings of Theorem \ref{thm:op}.
The main objective is to obtain \emph{qualitative}  results from the rather abstract expressions of the efficient frontier in Theorem \ref{thm:op}.


In the numerical analysis, we study a univariate example with $m = k =1$ (i.e., there is only one risky asset whose price jumps are modeled by a one-dimensional Hawkes process).
We follow \cite{ait2015modeling} and \cite{liu2021household} to assign the model parameters for this example as in Table \ref{tab:para_1d}.
\begin{table}[H]
	\centering
	\begin{tabular}{ccccccccccc} \hline
		$r$ & $\mu$ & $\sigma$ & $J$ & $\Eb[Z]$ & $\Eb[Z^2]$ & $\alpha$ & $\beta$ & $\lam_\infty$ & $T$ & $x_0$ \\\hline
		2\% & 9\%  & 20\% & 1 & -2\% & 6\% & 5 & 0.1 & 0.48 & 2 & 1\\ \hline
	\end{tabular}
	\caption{Model parameters in the numerical example}
	\label{tab:para_1d}
\end{table}

\subsection{The solution to (\ref{eq:PDE-g})}
\label{sub:solution}

Recall from Theorem \ref{thm:op} that both the efficient strategy \eqref{eq:efficient-strategy} and the efficient frontier \eqref{eq:efficient-frontier} depend on $g(0, \lam_0)$, where $g(\cdot, \cdot)$  solves the non-local partial differential equation \eqref{eq:PDE-g} and $\lam_0 = \lam(0)$ is the initial intensity value of the Hawkes process $N$.
Equation \eqref{eq:PDE-g} does not have an analytic solution and is difficult to solve numerically as well, due to the presence of the non-local component.
In particular, Equation \eqref{eq:PDE-g} dictates that the partial derivatives of $g$ at $(t, \lam)$ relate to not only its value at the same point, $g(t,\lam)$, but also its value at the corresponding post-jump points, $g(t, \lam + \beta_{(l)})$, where $l=1,\ldots,m$. (Here in this example $\beta_{(l)} = \beta$ since $m=1$.)
In order to solve $g$ from \eqref{eq:PDE-g}, we apply an exponential transform and study $\widetilde{g}(t, \lam) = e^{ \, g(t, \lam)}$.
Based on the analysis in Appendix \ref{app:pde}, we obtain $\widetilde{g}(t, \lam)$ in \eqref{eq:wg}, which can be seen as the Feynman-Kac representation of $\widetilde{g}(t, \lam)$.
We then use this result to numerically solve for $\widetilde{g}(t, \lam)$ via Monte Carlo simulation, which is briefly discussed below.

Given $t \in [0,T]$ and $\lam \in [\underline{\lam}, \bar{\lam}]$, we choose time step $\Delta_t$ and space step $\Delta_\lam$, and partition $ [0,T] \times  [\underline{\lam}, \bar{\lam}]$ to obtain  discrete grid points $(t_i, \lam_j)$, where $i=1,2,\ldots, 1 + {\cal N}_T$ (i.e., ${\cal N}_T  = T / \Delta_t$)  and $j=1,2,\ldots, 1+ {\cal N}_\lam^{(i)}$ (e.g., ${\cal N}_\lam^{(1)}  = (\bar{\lam} - \underline{\lam}) /  \Delta_\lam$) with $\Delta_t$ and
$\Delta_\lam$ properly chosen so that ${\cal N}_T$ and ${\cal N}_\lam^{(i)}$ are integers.
Note that the spatial dimension $1 + {\cal N}_\lam^{(i)}$ varies as time moves forward; see Remark \ref{rmk:dim} for details.
 The pseudo-algorithm for computing  $\big(\widetilde{g}(t_1, \lam_j) \big)_{j=1,2,\ldots,{\cal N}_\lam^{(1)}}$ is given below:
\begin{enumerate}
	\item Initialize $\widetilde{g}( t_{1 + {\cal N}_T}, \, \lam_j ) = 1$ for all $j = 1, 2,\ldots, {\cal N}_\lam^{(1 + {\cal N}_T)}$, since $\widetilde{g}(T, \lam) = 1$ for all $\lam > 0$.
	
	\item With all $\big\{ \widetilde{g}(t_k, \cdot) \big\}_{k \ge i + 1}$ computed, we calculate $\big\{ \widetilde{g}(t_i, \cdot) \big\}$ at time $t_i$ by the following two steps:
	
	(i) Simulate $M$ different paths for the intensity process between $t_i$ and $T$  via \eqref{eq:intensity} starting with $\lam(t_i) = \lam_j$, denoted by $\{\lam^{(\mathfrak{m})}(t_k): k = i, i+1, \ldots, 1 + {\cal N}_T \text{ with } \lam(t_i) = \lam_j \}$, where $\mathfrak{m}=1,\ldots,M$;

	(ii) For each path, compute the right hand side of \eqref{eq:wg}, denoted by $\widetilde{g}^{(\mathfrak{m})}(t_i, \lam_j)$, and obtain  $\widetilde{g}(t_i, \lam_j)= \frac{1}{M} \sum_{\mathfrak{m}=1}^M \, \widetilde{g}^{(\mathfrak{m})}(t_i, \lam_j)$.
\end{enumerate}

With the above algorithm in hand, we set  $T=2$, $\underline{\lam} = 0.1$ and $\bar{\lam} = 2$ (the rest are given in Table \ref{tab:para_1d}), and plot  $\widetilde{g}$ over $[0,2] \times [0.1, 2]$ in Figure \ref{fig:g_1d}.
It can be seen from Figure \ref{fig:g_1d} that $\widetilde{g}$ (or $g$) is an increasing function of both time argument $t$ and  spatial argument $\lam$.
Such an observation in comparative statics is consistent with the result derived from the Poisson model in Corollary \ref{cor:Poi}.
To be precise, the corresponding $g$ under the Poisson model is obtained in \eqref{eq:g_det}, from which we can easily see that $g$ is increasing with respect to $t$.
By writing $\lam_P(s) = \lam + \int_t^s \, \lam'(u) \, \dd u$ (here we assume the deterministic intensity $\lam_P(s)$ is differentiable over $[0,T]$) and using the definition of $\Gamma_P$ in \eqref{eq:Ga_P}, we obtain that $\Gamma_P$ is an increasing function of $\lam = \lam_P(t)$ and so is $g$.

\begin{figure}[h]
\centering
\includegraphics[trim=1cm 1cm 1cm 1cm, clip=true, width= 0.6 \textwidth]{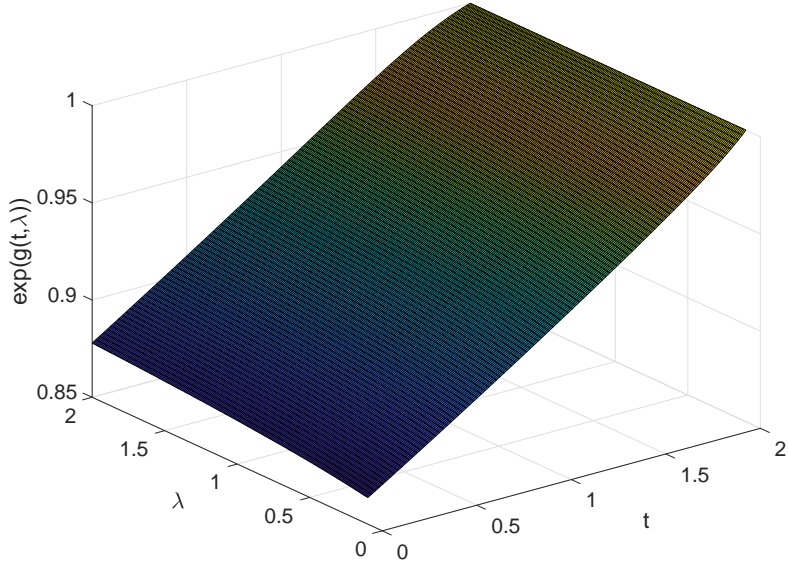}
\\[-2ex]
\caption{Graph of $\widetilde{g}(t, \lam) = e^{ \, g(t, \lam)}$ over $(t, \lam) \in [0,2] \times [0.1, 2]$}
\label{fig:g_1d}
\end{figure}

\begin{rmk}
	\label{rmk:dim}
We have mentioned in the above algorithm that the dimension of spatial points \emph{varies} along the time grid, and here we explain why this is the case for the non-local partial differential equation \eqref{eq:PDE-g}.
Given $\Delta_\lam$,  the spatial dimension at time $t_1$ ($t_1 = 0$) is easily obtained via $1 + {\cal N}_\lam^{(1)} =  1 + (\bar \lam - \underline \lam) / \Delta_\lam$.
Now we investigate what spatial points must be included at time $t_2$ ($t_2 = \Delta_t$) in order to calculate $\widetilde{g}(t_1, \lam(t_1) )$, where $\lam(t_1) \in [\underline \lam, \bar \lam]$.
To compute $\widetilde{g}(t_1, \lam(t_1))$ via \eqref{eq:wg}, we need $\widetilde{g}(t_2, \lam(t_2))$ and $\widetilde{g}(t_2, \lam(t_2) + \beta)$ for all possible $\lam(t_2)$ at $t_2$ that can be reached from $\lam(t_1)$.
Given any $\lam(t_1) \in [\underline \lam, \bar \lam]$ on the spatial grid, by discretization of \eqref{eq:intensity}, $\lam(t_2)$ takes two possible values: (1)  $\lam(t_2) = \lam(t_1) + \alpha ( \lam_\infty - \lam(t_1)) \Delta_t$ (no jump case); and (2) $\lam(t_2) = \lam(t_1) + \alpha ( \lam_\infty - \lam(t_1)) \Delta_t +  \beta$ (jump case).
In addition, we also need the values $\widetilde{g}(t_2, \lam(t_2) + \beta)$, which are evaluated at $\lam(t_2) + \beta$, i.e., the post-jump point of $\lam(t_2)$.
In summary, not only the dynamics of the intensity process and but also the non-local equation \eqref{eq:PDE-g} (see $U$ in \eqref{eq:U(b)}) involve jumps, both of which may lead to the change of the spatial dimension along time.
\end{rmk}

\subsection{Sensitivity analysis}

In this subsection, we focus on the sensitivity analysis of the efficient frontier  under the Hawkes-jump-diffusion model (hereafter Hawkes model), which is given by  \eqref{eq:efficient-frontier} in Theorem \ref{thm:op}.
We are particularly interested in how the parameters of the Hawkes intensity process affect the efficient frontier.
In all the analysis below, we only allow one parameter to vary at different levels and fix the rest as in Table \ref{tab:para_1d}.

We first investigate how the initial intensity value $\lam_0$ and the mean-reversion level $\lam_\infty$ affect the efficient frontier under the Hawkes model (see \eqref{eq:vec} for their definitions). The results are plotted in Figure \ref{fig:EF_lam}.
In the left panel of Figure \ref{fig:EF_lam}, we consider three levels for the initial intensity $\lam_0$: (1) low level $\lam_0 = 0.1$; (2) mean-reversion level $\lam_0 = \lam_\infty = 0.48$; and (3) high level $\lam_0 = 1.9$.
We observe from the graphs that the efficient frontier moves downward (i.e., it deteriorates) as $\lam_0$ increases.
To see this result, notice when the current intensity level increases, all the subsequent intensity increases as well (although the impact of the current increment decays exponentially), which makes the MV investor worse off.
In the right panel of Figure \ref{fig:EF_lam}, we fix $\lam_0 = \lam_\infty$ and consider three levels for the mean-reversion level $\lam_\infty$: $\lam_\infty= 0.3, 0.48, 0.8$.
The graphs show that the MV investor benefits when $\lam_\infty$ decreases.
To understand this finding, note that the higher the $\lam_\infty$, the higher the intensity $\lam$ on average, as $\lam$ mean reverts to $\lam_\infty$.
A comparison between the two panels indicates that $\lam_\infty$ has a more significant impact on the efficient frontier than $\lam_0$.

\begin{figure}[htb]
	\centering
	\includegraphics[trim=1cm 0cm 1cm 1cm, clip=true, width= 0.48 \textwidth]{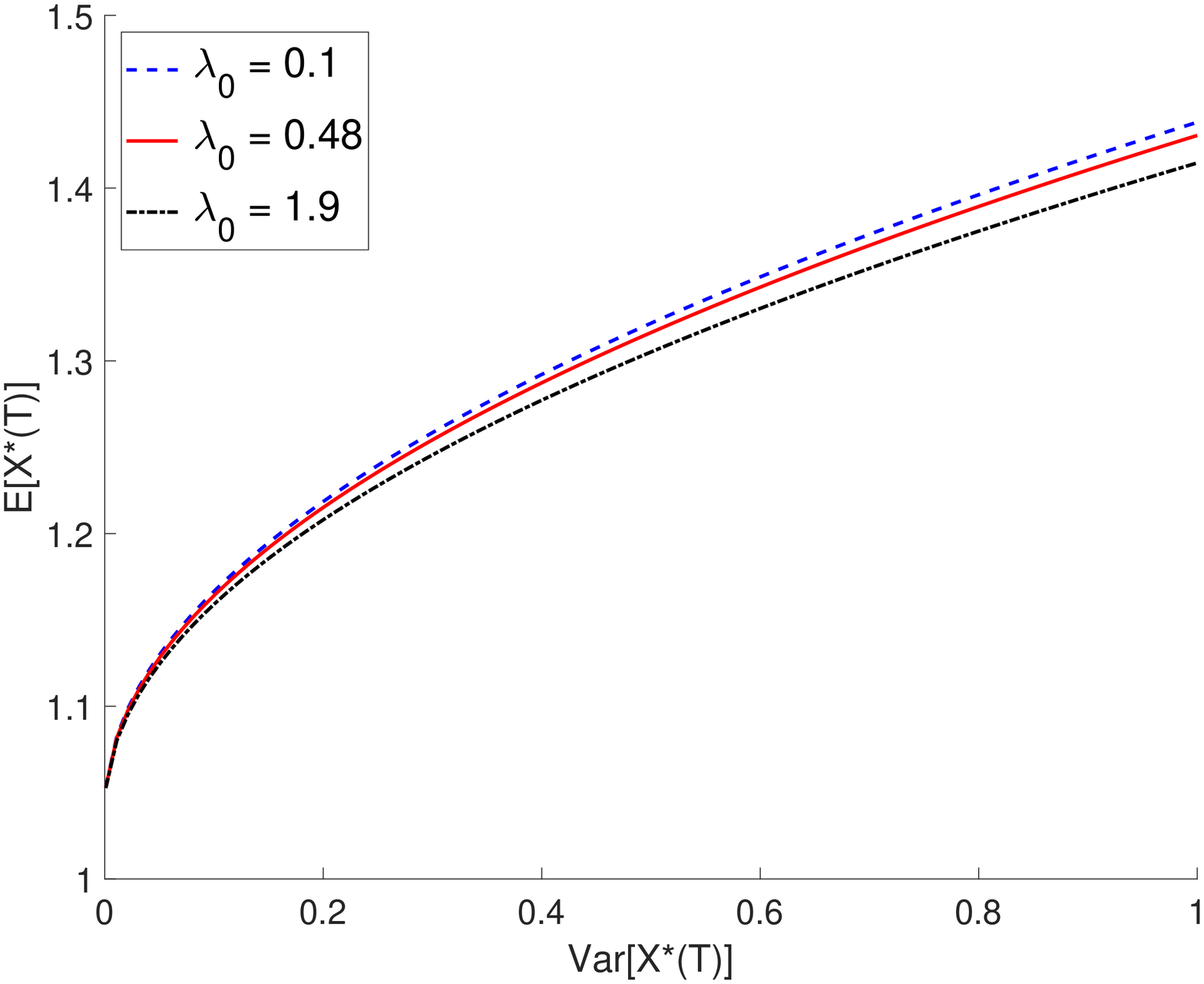}
	\includegraphics[trim=1cm 0cm 1cm 1cm, clip=true, width= 0.48 \textwidth]{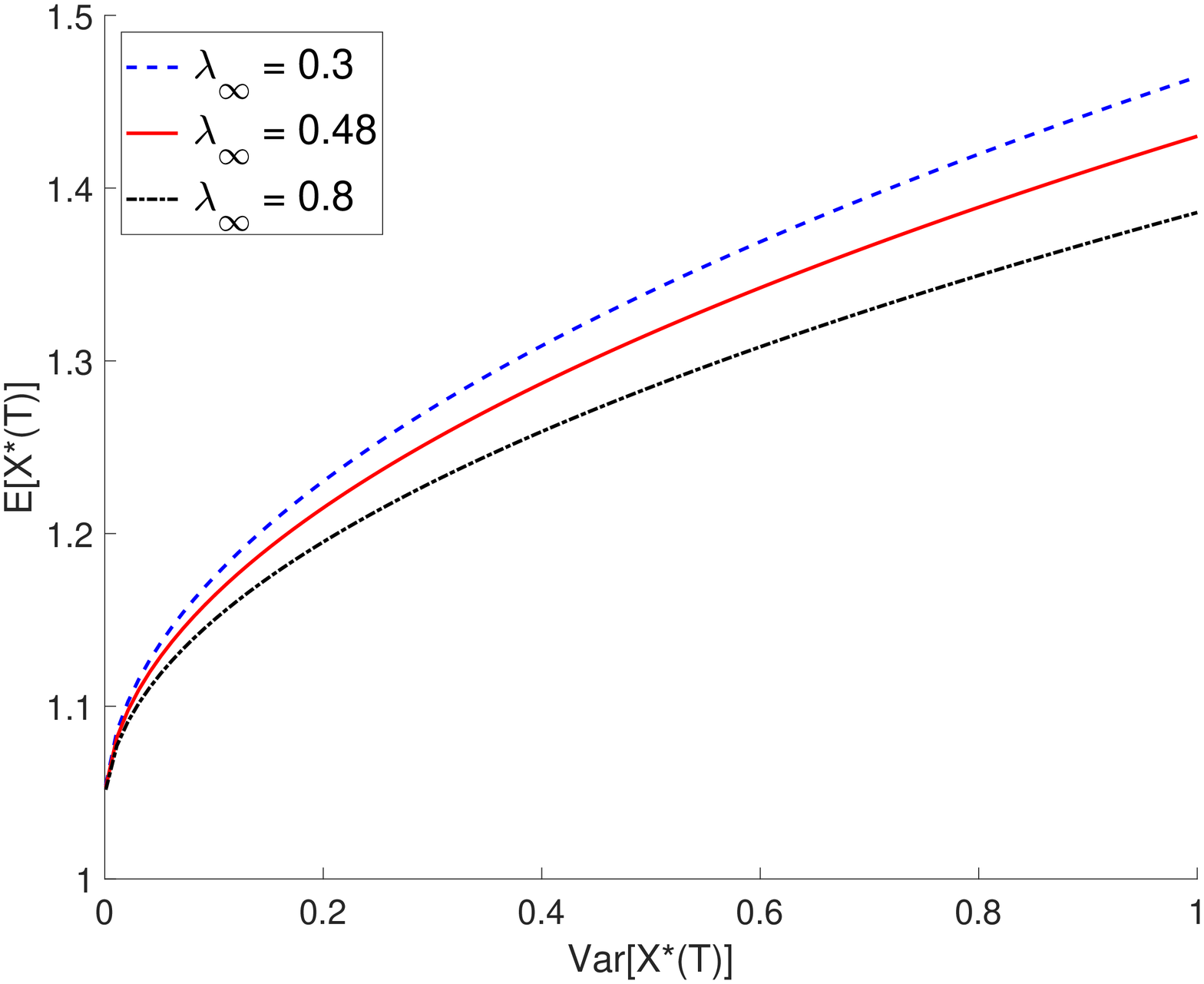}
	\\[-3ex]
	\caption{Impact of $\lam_0$ (Left) and $\lam_\infty$ (Right) on the efficient frontier under the Hawkes model}
	\label{fig:EF_lam}
\end{figure}

We next study the impact of $\alpha$ and $\beta$ on the efficient frontier. Recall from \eqref{eq:intensity} that $\alpha$ is the mean-reversion speed and $\beta$ is the jump size of the intensity process $\lam$.
We plot the results in Figure \ref{fig:EF_1d}.
In the left panel of Figure \ref{fig:EF_1d}, we consider three different levels for $\alpha$: $\alpha = 1, 2, 5$. When $\alpha$ increases, the intensity process mean reverts to  $\lam_\infty$ faster, which improves the efficient frontier.
The reason for such a result is that with a higher $\alpha$, the intensity $\lam$ spends less time at an abnormal level and is more likely to be around $\lam_\infty$, which reduces the uncertainty of the risky asset.
In the right panel of Figure \ref{fig:EF_1d}, we consider three different levels for $\beta$: $\beta = 0.1, 0.5, 2$.
It is readily seen that when $\beta$ increases, the efficient frontier worsens.
Recall $\beta$ is the jump size of the intensity $\lam$ upon a jump of the Hawkes process $N$.
Therefore, given a bigger $\beta$, the increment of $\lam$ upon a jump of $N$ is larger, which further ``excites'' more jumps in the near future and increases the variance in the efficient strategy.
For both panels in Figure \ref{fig:EF_1d}, we set $\lam_0 = 1$, so that the difference among the three efficient frontiers is more visible. We emphasize that all the findings remain the same for a smaller $\lam_0$.

\begin{figure}[htb]
	\centering
	\includegraphics[trim=1cm 0cm 1cm 1cm, clip=true, width= 0.48 \textwidth]{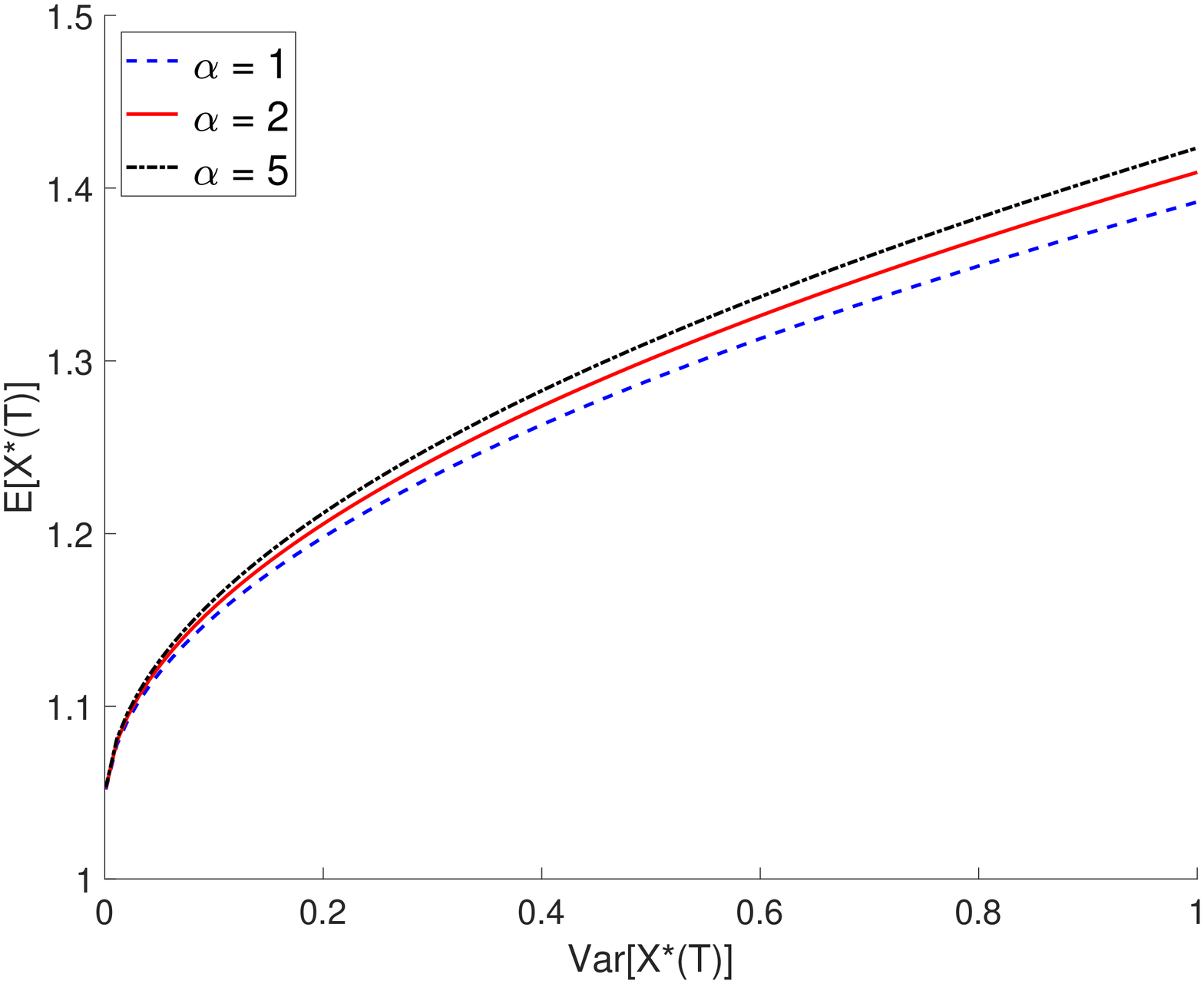}
	\includegraphics[trim=1cm 0cm 1cm 1cm, clip=true, width= 0.48 \textwidth]{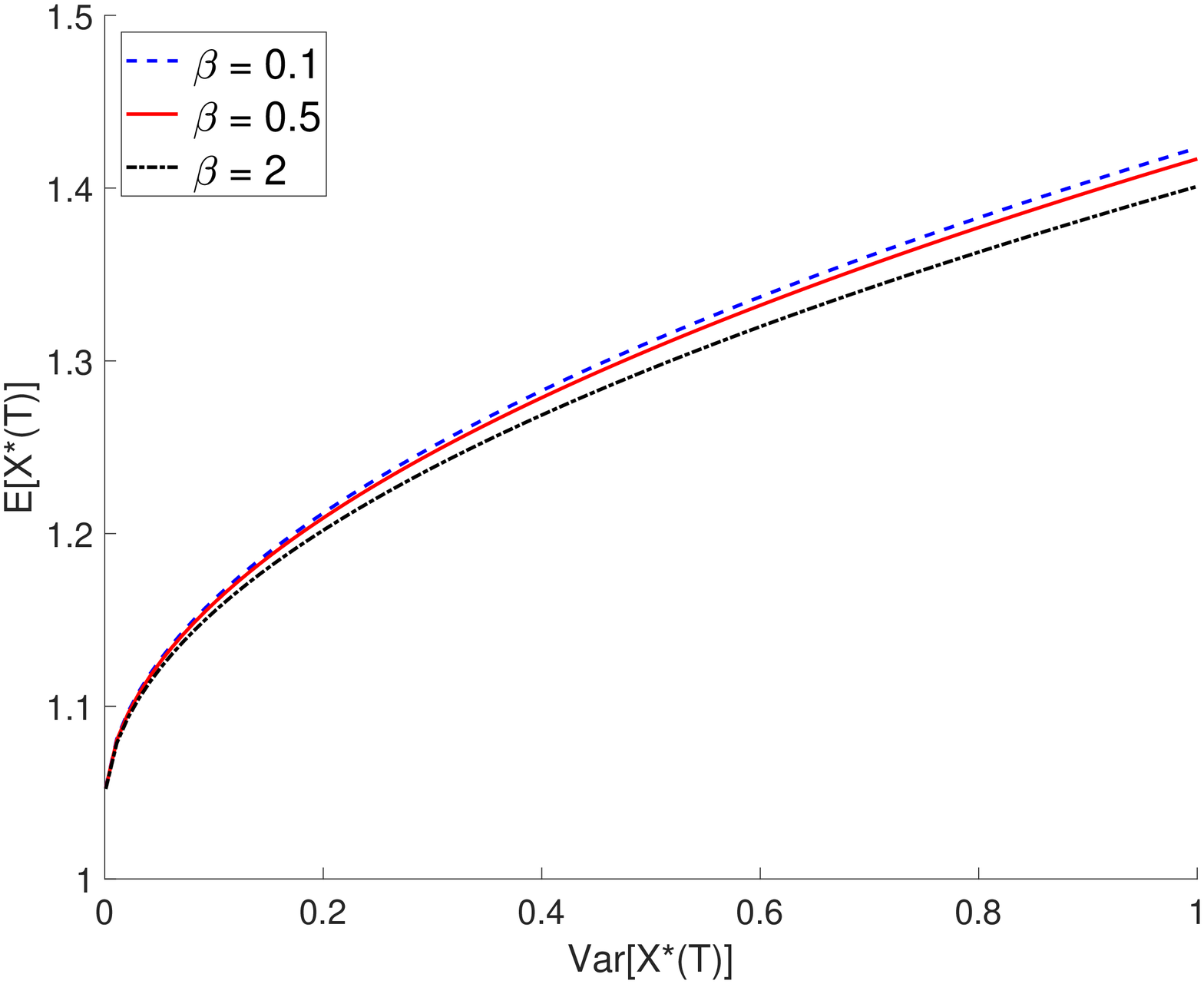}
	\\[-3ex]
	\caption{Impact of $\alpha$ (Left) and $\beta$ (Right) on the efficient frontier under the Hawkes model}
	\label{fig:EF_1d}
\end{figure}

\subsection{Comparisons with the Poisson model}
\label{sub:comp_Poi}

In the last part of the numerical analysis, we conduct a comparison study between the Hawkes model \eqref{eq:dS_vec} and the Poisson-jump-diffusion model (hereafter, Poisson model) \eqref{eq:dS_vec-Poisson}.
In the first study, we consider a relatively small value for the initial intensity ($\lam_0 = 0.3$) in the Hawkes model, and three different levels for the \emph{constant} intensity ($\lam_P = 0.3, 0.4, 0.5$) in the Poisson model.
The corresponding efficient frontiers are plotted in the left panel of Figure \ref{fig:EF_com_1d}.
In the second study, we consider a relatively large value for the initial intensity ($\lam_0 = 0.7$) in the Hawkes model, and three different levels for the \emph{constant} intensity ($\lam_P = 0.4, 0.7, 0.8$) in the Poisson model.
The corresponding efficient frontiers are plotted in the right panel of Figure \ref{fig:EF_com_1d}.
We observe from Figure \ref{fig:EF_com_1d} that if $\lam_0 < \lam_\infty = 0.48$ (resp. $\lam_0 > \lam_\infty = 0.48$),
the Hawkes model with the initial intensity $\lam_0$ yields a worse (resp. better) efficient frontier than the Poisson model with the same constant intensity $\lam_P = \lam_0$.
To understand this finding, let us recall from Table \ref{tab:para_1d} that $\lam_\infty = 0.48$ (mean-reversion level of the Hawkes intensity $\lam$), $\alpha = 5$ (mean-reversion speed), $\beta = 0.1$ (jump size) and $T=2$ (investment horizon).
In consequence, under the given parameters,  the mean-reversion effect dominates the jump effect in the dynamics of the Hawkes intensity process $\lam$ given by \eqref{eq:dlam_l}.
Therefore, if $\lam_0 < \lam_\infty$ (resp. $\lam_0 > \lam_\infty$), setting $\lam_P = \lam_0$ underestimates (resp. overestimates) the overall jump intensity, while the increase of the (stochastic or deterministic) jump intensity deteriorates the efficient frontier (see both Figures \ref{fig:EF_1d} and \ref{fig:EF_com_1d}), which together explain the comparative finding of Figure \ref{fig:EF_com_1d}.

\begin{figure}[htb]
	\centering
	\includegraphics[trim=1cm 0cm 1cm 1cm, clip=true, width= 0.48 \textwidth]{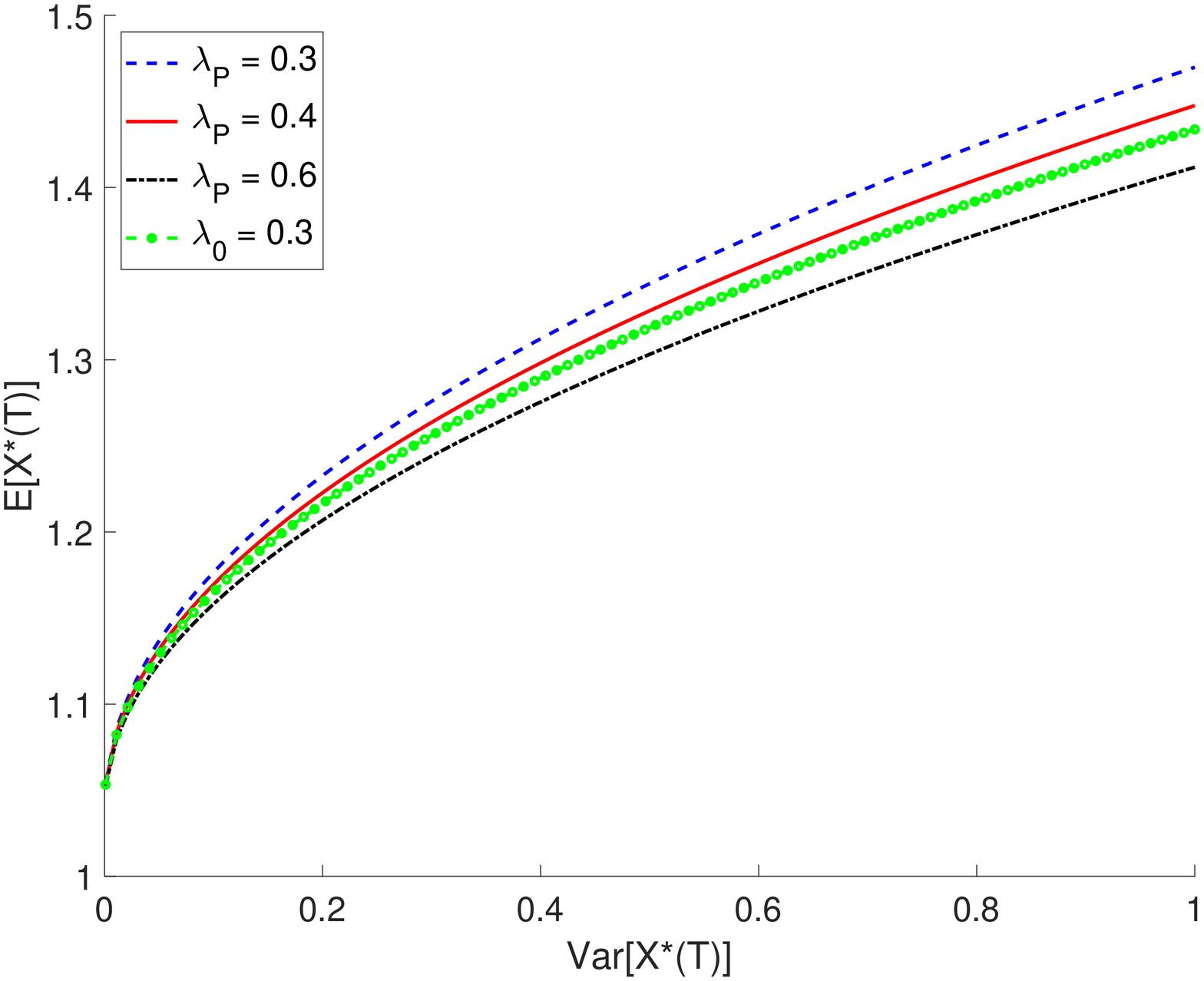}
	\includegraphics[trim=1cm 0cm 1cm 1cm, clip=true, width= 0.48 \textwidth]{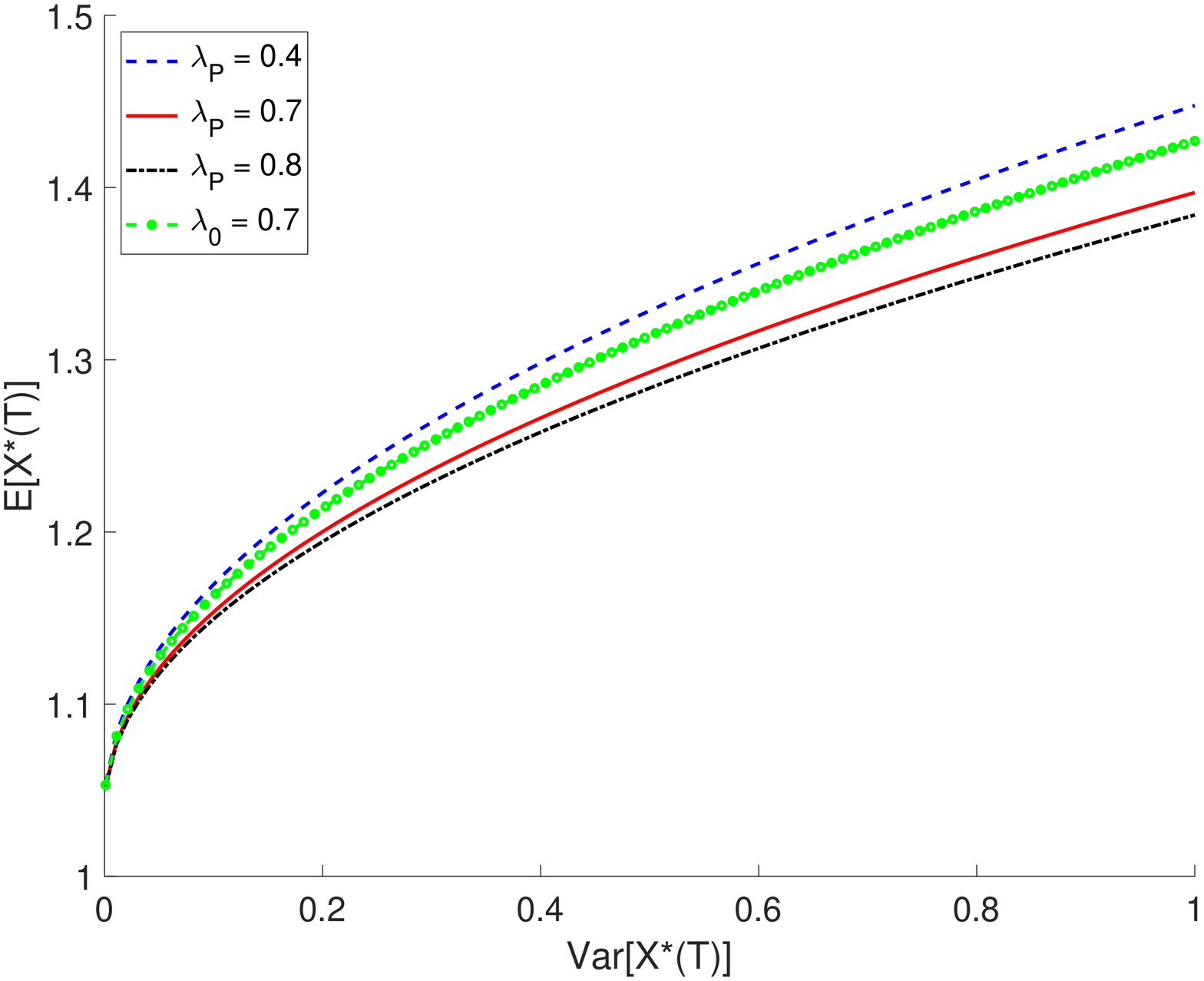}
	\\[-3ex]
	\caption{Comparisons of the efficient frontiers under the Hawkes and Poisson models ($\beta = 0.1$)}
	\label{fig:EF_com_1d}
\end{figure}

When a more subtle analysis of Figure \ref{fig:EF_com_1d} is performed, we find that the efficient frontier of the Hawkes model with $\lam_0 = 0.3$ (resp. 0.7) is very close to that of the Poisson model with $\lam_P = 0.47$ (resp. 0.51).
Again recall from Table \ref{tab:para_1d} that $\lam_\infty = 0.48$.
It then seems necessary to investigate the case of $\lam_0 = \lam_\infty = 0.48$ to further compare the Hawkes model with the Poisson model.
To see how the self-excitation effect of the Hawkes process affects the efficient frontier, we consider a larger jump size with $\beta = 2$. (Each jump of the Hawkes process $\dd N_t = 1$ ``excites'' its own intensity by $\beta$.) We plot the results in Figure \ref{fig:EF_com_1d_inf}.
A clear message from Figure \ref{fig:EF_com_1d_inf} is that the self-excitation feature of Hawkes processes increases the 
variance of the terminal wealth under the efficient strategy (for any chosen expectation target $\xi \ge x_0 e^{rT}$)
and thus leads to a worsened efficient frontier, comparing to the Poisson model with intensity $\lambda_P = \lambda_0 = \lambda_\infty$, where $\lambda_0$ (resp. $\lambda_\infty$) is the initial value (resp. long-term value) of the Hawkes intensity process.
(We comment that the finding of Figure \ref{fig:EF_com_1d_inf} remains the same when $\beta = 0.1$ from Table \ref{tab:para_1d} is used, though one needs to zoom in to observe the difference between the two efficient frontiers.)

\begin{figure}[htb]
\begin{center}
	\includegraphics[trim=1cm 0cm 1cm 1cm, clip=true, width= 0.48 \textwidth]{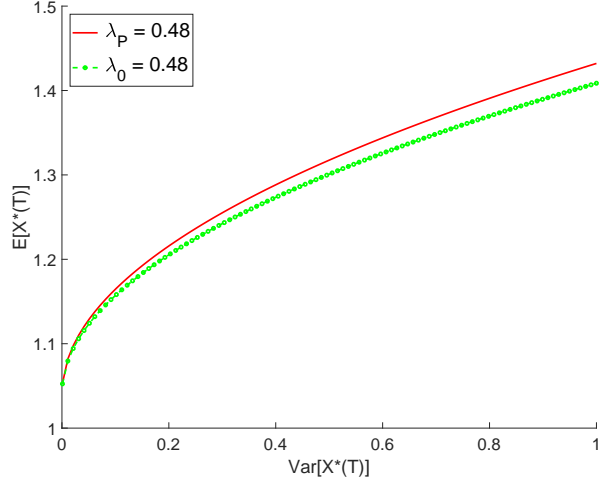}
	\\[-3ex]
	\caption{Comparisons of the efficient frontiers under the Hawkes and Poisson models ($\beta = 2$)}
	\label{fig:EF_com_1d_inf}
\end{center}
\end{figure}

\begin{rmk}
	In this remark, we offer some discussions on the difficulty of solving \eqref{eq:PDE-g} numerically when the spatial dimension is greater than 1 (often known as the ``curse of dimensionality'').
	For this purpose, we first review the computational complexity of the one-dimensional case ($m=k=1$) in Subsections \ref{sub:solution}-\ref{sub:comp_Poi}. To solve \eqref{eq:PDE-g} numerically in the one-dimensional case,  we choose $\Delta_t = 0.02$ and $T=2$, implying the time dimension is 101, and 	$\Delta_\lam = 0.01$ and $\lam_0 \in [0.1, 2]$, implying the spatial dimension at time 0 is 191. Recall from Remark \ref{rmk:dim} that the spatial dimension is varying as time evolves, and the maximum spatial dimension number is 239 over the total 101 time grid points, corresponding to $\max \lam = 2.48$, under the given parameters. The computation algorithm moves backward in time with the terminal values $\widetilde{g}(T=t_{101}, \lam_j) = 1$, and, at each grid point $(t_i, \lam_j)$, simulates $M$ (chosen to be 5,000) paths to compute $\widetilde{g}(t_i, \lam_j)$, yielding a complexity of approximately $100 \times 235 \times 5000 = 1.175 \times 10^8$ (the mean of the spatial dimensions is about 235). We also mention that computing the integrand in \eqref{eq:wg} is also time consuming when time $t$ is getting close to 0.
	
	Now we are ready to discuss the simplest high dimensional case with $m=2$ and $k=2$ (i.e., there are two risky assets whose jumps are derived from a bivariate Hawkes process). Assume the self-excitation effect parameters $\beta_{11} = \beta_{22} = 0.1$ (i.e., a jump of $N_i(t)$ at time $t$  ``excites'' its own intensity $\lam_i(t)$ by 0.1, where $i=1,2$) and the cross-excitation effect parameters $\beta_{12} = \beta_{21} = 0.08$ (i.e., a jump of $N_i(t)$ at time $t$ ``excites'' the other's intensity $\lam_j(t)$ by 0.08). Due to the additional cross-excitation effect, the mean of the spatial dimensions in the bivariate case is about 300 (recall the number is 235 in the univariate case). As such, the total number of simulations quickly grows to $100 \times 300^2 \times 5000 = 4.5 \times 10^{10}$.
	In addition, computing the integrand in \eqref{eq:wg} now involves finding the inverse of a matrix within each time step $\Delta_t$.
	In fact, the computation is so complex that we cannot even obtain $\widetilde{g}(0, \lam_0)$ for a single set of parameter configuration in three days on a personal laptop (MacBook Pro 2017 with 3.1G Dual Core i5 CPU, 8GB memory, and 512G SSD).
	This is not entirely surprising though, as solving a multi-dimensional partial differential equation with a non-local component remains an open question to our knowledge, and there are no universally efficient numerical methods to tackle such a problem.
\end{rmk}

\section{Conclusion}
\label{sec:con}

In this paper, we consider a mean-variance portfolio selection problem in a contagious financial market, where the prices of risky assets are subject to mutually exciting jumps of a multivariate Hawkes process.
The contagion risk is then captured by the fact that a price jump of an asset will increase the jump intensities of both the same asset and all the other assets in the market.
\cite{ait2016portfolio} study portfolio selection problems under the utility maximization criterion in a similar market model with contagion risk, and apply the standard HJB equation method to obtain optimal investment strategies.
Here we obtain semi-explicit solutions to the efficient strategy and the efficient frontier via the stochastic maximum principle, BSDE theory, and LQ control technique. Our paper is among one of the few papers concerning optimal control of dynamical systems with Hawkes-type jumps in the literature of mathematical finance, and we hope that our work could motivate further research in this area. An interesting future research direction is to consider non-Markovian Hawkes processes with non-exponential decay functions. In that case, the HJB equation approach will not be applicable,
while the stochastic maximum principle together with BSDEs still works.

\section*{Acknowledgments}
We would like to thank anonymous associate editor and referees for their careful reading and many insightful comments that help us improve the quality of an early version of this paper. Bin Zou is partially supported by a start-up grant from the University of Connecticut. Yang Shen is partially supported by the Discovery
Early Career Researcher Award (No. DE200101266) from the Australian Research Council, the National Natural Science Foundation of China under Grant (No. 71771220), and the Major Program of the National Social Science Foundation of China  (No. 18ZDA092).

\appendix
\section{Proof to Lemma \ref{lem:g}}
\label{app:lem}

\begin{proof}[Proof to Lemma \ref{lem:g}]
We start by rewriting the dynamics of $Y$ in \eqref{eq:BSDE-(Y,V)-comp} by
\begin{align}
	\dd Y (t) &= Y(t) \left( - 2 r  + \wz(t, \lam (t))^\top \Gamma (t, \lam  (t))^{-1}\wz(t, \lam (t)) \right) \dd t
	+ Y (t-) U (t, \lambda (t))^\top \dd {\widetilde N} (t),
\end{align}
which implies that $\{e^{2 r t} Y(t) \}_{t \in [0,T]}$ is a sub-martingale, and thus
\begin{align}
	e^{2 r t} Y (t)  \leq {\mathbb E}_t \bigg [ e^{2 r T} Y (T)  \bigg ],\qquad \forall \, t \in [0,T] .
\end{align}
Recall $Y(T)=2$ and \eqref{eq:Y-ansatz}, we have
\begin{align}
	2 e^{2 r (T-t)} \cdot e^{g (t, \lam (t))} = Y (t) \leq 2 e^{2 r (T-t)} , \quad \forall t \in [0, T]
\end{align}
and thereby
\begin{align}\label{eq:e^g-bound}
	0 \leq e^{g (t, \lam (t))} \leq 1 , \quad \forall t \in [0, T] .
\end{align}
Moreover, the lower bound of $Y$ established in \eqref{eq:Y-lowerbound} leads to a strictly positive lower bound for $e^{g (t, \lambda (t))}$, that is,
\begin{align}\label{eq:e^g-lowerbound}
e^{g (t, \lambda (t))} \geq \frac{\delta}{2} e^{-2 r T} := \epsilon , \quad \forall t \in [0, T] .
\end{align}
This completes the proof of $\epsilon \leq e^{g (t, \lambda (t))} \leq 1$, for any $t \in [0, T]$.

Next we focus on the case $t=0$. Note that \eqref{eq:e^g-lowerbound} holds for any $t \in [0, T]$. Hence, it remains to show that $e^{g (0, \lam_0)} < 1$. To that end, we recall \eqref{eq:PDE-g} and apply It\^o's formula to $e^{g (t, \lam (t))}$:
\begin{align}
	\dd e^{g (t, \lam (t))} &= e^{g (t, \lam (t))} \big [ g_t(t, \lam (t))
	+ g_\lambda^\top(t,\lam (t)) \, \alpha (\lambda_{\infty} - \lambda (t) )
	+ U (t, \lambda (t))^\top \lambda (t) \big ] \dd t
	+ e^{g (t, \lam (t))} U (t, \lambda (t))^\top \dd {\widetilde N} (t) \\
	&= e^{g (t, \lam (t))} \wz(t, \lambda (t))^\top \Gamma (t, \lambda (t))^{-1} \wz(t, \lambda (t)) \dd t
	+ e^{g (t, \lam (t))} U (t, \lambda (t))^\top \dd {\widetilde N} (t) .
\end{align}
Since $e^{g (t, \lam (t))} U (t, \lambda (t))^\top$ is bounded, the jump component in the above equation is a martingale. Thus,
integrating from $t_1$ to $t_2$ and conditioning on $t_1$, where $0 \leq t_1 \leq t_2 \leq T$, we have
\begin{align}\label{eq:e^g}
	0 \leq e^{g (t_2, \lam (t_2))} - e^{g (t_1, \lam (t_1))}
	= {\mathbb E}_{t_1} \bigg [ \int^{t_2}_{t_1} e^{g (s, \lam (s))} \wz(s, \lambda (s))^\top \Gamma (s, \lambda (s))^{-1}
	\wz(s, \lambda (s)) \dd s \bigg ] \leq 1 .
\end{align}
Hence, we have $e^{g (t, \lam (t))} \geq e^{g (0, \lam_0)}$, for any $t \in [0, T]$. If $e^{g (0, \lam_0)} = 1$,
it then follows from \eqref{eq:e^g-bound} that $e^{g (t, \lam (t))} \equiv 1$, for any $t \in [0, T]$. In that case, from \eqref{eq:e^g},
we obtain
\begin{align}
	{\mathbb E} \bigg [ \int^{T}_{0} \wz(s, \lambda (s))^\top \Gamma (s, \lambda (s))^{-1}
	\wz(s, \lambda (s)) \dd s \bigg ] = 0 .
\end{align}
On the other hand, if $e^{g (t, \lam (t))} \equiv 1$, then
\begin{align}
	U (t, \lam (t)) \equiv - {\bf 1}_m =(-1, \ldots, -1)^\top \in \Rb^m , \quad \Gamma (t, \lam (t)) \equiv \sigma \sigma^\top , \qquad \forall t \in[0, T] .
\end{align}
Recall that the non-degeneracy of the precision matrix $(\sigma \sigma^\top)^{-1}$ is postulated. The above results lead to
\begin{align}
	\wz(t, \lambda (t)) = B
	- \int_{(-1, \infty)^m} \eta (z)\, \mathrm{Diag}[\lam (t) \bullet \nu (\dd z)] \, {\bf 1}_m = {\bf 0}_k , \qquad \mathbb{P}\mbox{-a.s.}, \quad \forall t \in [0, T] ,
\end{align}
which cannot hold for all $t$ since $\lam$ is a stochastic  process.
Therefore, $e^{g (0, \lam_0)} < 1$ must hold true.
\end{proof}

\section{Proof of Lemma \ref{lem:g-solvability}}
\label{app:pde}

\begin{proof}[Proof to Lemma \ref{lem:g-solvability}]

To analyze \eqref{eq:PDE-g}, we apply exponential transformation to study ${\widetilde g} (t, \lambda) : = e^{g (t, \lambda)}$.
Recalling \eqref{eq:e^g}, we have
\begin{align}
{\widetilde g} (t, \lambda) = \ e^{g (t, \lam)}
	=& \ 1- {\mathbb E}_{t,\lam} \bigg [ \int^{T}_{t} e^{g (s, \lam (s))} \wz(s, \lambda (s))^\top \Gamma (s, \lambda (s))^{-1}
	\wz(s, \lambda (s)) \, \dd s \bigg ] \\
=& \ 1- {\mathbb E}_{t,\lam} \bigg [ \int^{T}_{t} {\widetilde g} (s, \lam (s)) \wz^{\widetilde g}(s, \lambda (s))^\top \Gamma^{\widetilde g} (s, \lambda (s))^{-1}
	\wz^{\widetilde g} (s, \lambda (s)) \, \dd s \bigg ] , \label{eq:wg}
\end{align}
where we use $\Eb_{t, \lam}$, instead of $\Eb_t$, to emphasize that the expectation is take under the condition $\lam(t) = \lam$, and the functions in the integrand are defined by
\begin{align}
 \Gamma^{\widetilde g} (t, \lambda) :=& \ \sig \sig^\top + \int_{(-1, \infty)^m} \; \eta(z) \, \mathrm{Diag}[(U^{\widetilde g} (t, \lambda) + \bm{1}_m) \bullet \lam \bullet \nu(\dd z)] \, \eta(z)^\top,
\label{eq:Gamma(t,b)_wg} \\
\wz^{\widetilde g} (t, \lambda) :=& \ B + \int_{(-1, \infty)^m} \; \eta(z) \, \mathrm{Diag}[\lam \bullet \nu(\dd z)] \, U^{\widetilde g} (t, \lambda),
\label{eq:Z_hat_wg} \\
U^{\widetilde g} (t, \lambda) :=& \bigg ( \frac{{\widetilde g} (t, \, \lambda +\beta_{(1)})}{{\widetilde g} (t, \, \lambda)} - 1, \cdots, \frac{{\widetilde g} (t, \, \lambda +\beta_{(m)})}{{\widetilde g} (t, \, \lambda)} - 1 \bigg )^\top. \label{eq:U(b)_wg}
\end{align}

Next we apply Schauder's fixed point theorem to prove the existence of a solution to \eqref{eq:wg}.
To that end, we consider a Banach space ${\cal X} := {\mathcal S}_{{\cal F}, L} (0, T; \mathbb R)$
equipped with the sup norm $\| \cdot \|_{\cal X}$
\begin{align}
\| {\widetilde g} \|_{\cal X} : = {\mathbb E} \bigg [ \sup_{t \in [0, T]} e^{- L (T-t)} \big| {\widetilde g} (t, \lam (t)) \big| \bigg ] .
\end{align}
where $L$ is a positive constant.

Define a map ${\cal T}$ from ${\cal M} := {\mathcal S}_{{\cal F}, L} (0, T; [\epsilon,1])$ onto itself as follows:
\begin{align}
({\cal T} {\widetilde g}) (t)  : = 1- {\mathbb E}_{t, \lam} \bigg [ \int^{T}_{t} {\widetilde g} (s, \lam (s)) \wz^{\widetilde g}(s, \lambda (s))^\top \Gamma^{\widetilde g} (s, \lambda (s))^{-1}
	\wz^{\widetilde g} (s, \lambda (s)) \, \dd s \bigg ] .
\end{align}
Obviously, ${\cal M} = {\mathcal S}_{{\cal F}, L} (0, T; [\epsilon,1])$ is a non-empty, convex, and compact space.

In what follows, we show that the map ${\cal T}: {\cal M} \subset {\cal X} \rightarrow {\cal M}$ is continuous.
For such a purpose, let us choose any ${\widetilde g}_1, {\widetilde g}_2 \in {\cal M}$. Then, we have
{\footnotesize
	\begin{align}
		& \ \| ({\cal T} {\widetilde g}_1) - ({\cal T} {\widetilde g}_2) \|_{\cal X} \\
		& = {\mathbb E} \Bigg \{ \sup_{t \in [0, T]} e^{- L (T-t)} \left| {\mathbb E}_{t,\lam} \bigg [ \int^{T}_{t} {\widetilde g}_1 (s, \lam (s)) \wz^{\widetilde g_1}(s, \lambda (s))^\top \Gamma^{\widetilde g_1} (s, \lambda (s))^{-1}
		\wz^{\widetilde g_1} (s, \lambda (s)) \dd s \right. \\
		& \left. \qquad\qquad\qquad\qquad\qquad - \int^{T}_{t} {\widetilde g}_2 (s, \lam (s)) \wz^{\widetilde g_2} (s, \lambda (s))^\top \Gamma^{\widetilde g_2} (s, \lambda (s))^{-1}
		\wz^{\widetilde g_2} (s, \lambda (s))\dd s \bigg ] \right| \Bigg \} \\
		& \leq K_1 {\mathbb E} \Bigg \{ \sup_{t \in [0, T]}  e^{- L (T-t)} {\mathbb E}_{t,\lam} \bigg [ \int^{T}_{t} \big | {\widetilde g}_1 (s, \lam (s)) - {\widetilde g}_2 (s, \lam (s)) \big |  | \lambda (s) |^2 \dd s \bigg ] \Bigg \} \\
		& \leq K_1 {\mathbb E} \Bigg \{ \sup_{t \in [0, T]}  e^{- L (T-t)} {\mathbb E}_{t,\lam} \bigg [
		\sup_{s \in [0, T]} e^{-L (T-s)} \big | {\widetilde g}_1 (s, \lam (s)) - {\widetilde g}_2 (s, \lam (s)) \big |
		\int^{T}_{t} e^{L (T-s)} | \lambda (s) |^2 d s \bigg ] \Bigg \} \\
		& \leq K_1 {\mathbb E} \Bigg \{ \sup_{t \in [0, T]}  e^{- L (T-t)} \bigg \{ {\mathbb E}_{t,\lam} \bigg [
		\sup_{s \in [0, T]} e^{-L (T-s)} \big | {\widetilde g}_1 (s, \lam (s)) - {\widetilde g}_2 (s, \lam (s)) \big |^2 \bigg ]
		\bigg \}^{\frac{1}{2}}
		\bigg \{ {\mathbb E}_{t,\lam} \bigg [ \bigg ( \int^{T}_{t} e^{L (T-s)} | \lambda (s) |^2 d s \bigg )^2 \bigg ] \bigg \}^{\frac{1}{2}} \Bigg \} \\
		& \leq K_1 {\mathbb E} \Bigg \{ \sup_{t \in [0, T]}  e^{- L (T-t)} \bigg \{ {\mathbb E}_{t,\lam} \bigg [
		\sup_{s \in [0, T]} e^{-L (T-s)} \big | {\widetilde g}_1 (s, \lam (s)) - {\widetilde g}_2 (s, \lam (s)) \big |^2 \bigg ]
		\bigg \}^{\frac{1}{2}}
		\bigg \{ {\mathbb E}_{t,\lam} \bigg [ \sup_{s \in [t, T]} | \lambda (s) |^4 \bigg ] \bigg \}^{\frac{1}{2}}
		\bigg ( \int^{T}_{t} e^{L (T-s)} d s \bigg ) \Bigg \} \\
		& \leq \frac{K_1}{L} (1- e^{- L T})  \Bigg \{ {\mathbb E} \bigg [ \sup_{t \in [0, T]} {\mathbb E}_{t,\lam} \bigg [
		\sup_{s \in [0, T]} e^{-L (T-s)} \big | {\widetilde g}_1 (s, \lam (s)) - {\widetilde g}_2 (s, \lam (s)) \big |^2 \bigg ]^\frac{\mathfrak{p}}{2} \bigg ] \Bigg  \}^{\frac{1}{\mathfrak{p}}} \bigg \{ {\mathbb E} \bigg [ \sup_{t \in [0, T]} {\mathbb E}_{t,\lam} \bigg [ \sup_{s \in [t, T]} | \lambda (s) |^4 \bigg ]^{\frac{\mathfrak{q}}{2}} \bigg ] \bigg \}^{\frac{1}{\mathfrak{q}}} ,
	\end{align}
}
where $\mathfrak{p}$ and $\mathfrak{q}$ are two positive constants such that $1/\mathfrak{p} + 1/\mathfrak{q} = 1$
and $\mathfrak{p} > 2$.

Next we apply the Doob martingale inequality and derive
\begin{align}
	& \Bigg \{ {\mathbb E} \bigg [ \sup_{t \in [0, T]} {\mathbb E}_{t, \lam} \bigg [
	\sup_{s \in [0, T]} e^{-L (T-s)} \big | {\widetilde g}_1 (s, \lam (s)) - {\widetilde g}_2 (s, \lam (s)) \big |^2 \bigg ]^\frac{\mathfrak{p}}{2} \bigg ] \Bigg  \}^{\frac{1}{\mathfrak{p}}} \\
	& \leq K_2 \Bigg \{ {\mathbb E} \bigg [ {\mathbb E}_{T} \bigg [
	\sup_{s \in [0, T]} e^{-L (T-s)} \big | {\widetilde g}_1 (s, \lam (s)) - {\widetilde g}_2 (s, \lam (s)) \big |^2 \bigg ]^\frac{\mathfrak{p}}{2} \bigg ] \Bigg  \}^{\frac{1}{\mathfrak{p}}} \\
	& \leq K_2 \Bigg \{ {\mathbb E} \bigg [
	\sup_{s \in [0, T]} e^{- L (T-s)} \big | {\widetilde g}_1 (s, \lam (s)) - {\widetilde g}_2 (s, \lam (s)) \big |^\mathfrak{p} \bigg ] \Bigg  \}^{\frac{1}{\mathfrak{p}}} \\
	& \leq K_2 \Bigg \{ {\mathbb E} \bigg [
	\sup_{s \in [0, T]} e^{- L (T-s)} \big | {\widetilde g}_1 (s, \lam (s)) - {\widetilde g}_2 (s, \lam (s)) \big | \bigg ] \Bigg  \}^{\frac{1}{\mathfrak{p}}} .
\end{align}
On the other hand, since the SDE \eqref{eq:intensity} for the Hawkes process satisfies the Lipschitz and linear growth conditions, from the moment estimates in the standard theory of SDEs (see, e.g., Theorem 1.9.4 in \citet{platen2010SDE}), we can derive
\begin{align}
	\bigg \{ {\mathbb E} \bigg [ \sup_{t \in [0, T]} {\mathbb E}_{t,\lam} \bigg [ \sup_{s \in [t, T]} | \lambda (s) |^4 \bigg ]^{\frac{\mathfrak{q}}{2}} \bigg ] \bigg \}^{\frac{1}{\mathfrak{q}}} &\leq K_3 \bigg \{ {\mathbb E} \bigg [ \sup_{t \in [0, T]}
	(1 + |\lam (t)|^4)^{\frac{\mathfrak{q}}{2}} \bigg ] \bigg \}^{\frac{1}{\mathfrak{q}}} \\
	&\leq K_3 \bigg \{ 1+ {\mathbb E} \bigg [ \sup_{t \in [0, T]}
	|\lam (t)|^4 \bigg ] \bigg \}^{\frac{1}{2}} \\
	&\leq K_3 \big \{ 1+ |\lam_0|^4 \big \}^{\frac{1}{2}} .
\end{align}
Note that in all the above derivations, $K_1$, $K_2$, and $K_3$ are generic positive constants and their values may vary from line to line.

Therefore, combining the above derivations, we have
\begin{align}
	\| ({\cal T} {\widetilde g}_1) - ({\cal T} {\widetilde g}_2) \|_{\cal X} \leq C \{ \| {\widetilde g}_1 - {\widetilde g}_2 \|_{\cal X} \}^{\frac{1}{\mathfrak{p}}} ,
\end{align}
where $C$ is a positive constant depending on $K_1$, $K_2$, $K_3$, $L$, and $\lam_0$.
This indicates that ${\cal T}$ is a continuous map. Using Schauder's fixed point theorem,
we can conclude that \eqref{eq:wg} admits a solution ${\widetilde g} \in \mathcal{S}_{{\cal F}, L}(0,T; [\epsilon, 1])$. From the relationship between
$g$ and ${\widetilde g}$, we can confirm that \eqref{eq:PDE-g} admits a solution $g : = \log [{\widetilde g}]$.


Next we show the uniqueness of the solution $\widetilde{g}$. For that purpose, recall that the quadratic-loss minimization problem
\eqref{pro:qlm} belongs to the class of stochastic LQ control problems, which admits a unique optimal control/strategy (see \cite{zhang2020lq}). Specifically, for \eqref{pro:qlm}, the unique optimal strategy $\pi^*_c$ is given by \eqref{eq:optimal-control}. As in Step 3 of the proof for Theorem \ref{thm:op}, we can show that (i) $\pi^*_c \in {\cal L}^2_{\cal F} (0, T; {\mathbb R}^k)$; (ii) the corresponding wealth process, denoted by ${\widehat X}^*$, is square-integrable
at the terminal time, i.e., ${\mathbb E} [({\widehat X}^* (T))^2] < \infty$.

Combining Assertion (i) and the uniqueness of $\pi^*_c$, we have that the wealth equation \eqref{eq:dX_hat} associated with $\pi^*_c$ admits a unique solution ${\widehat X}^* (\cdot) \in {\cal S}^2_{\cal F} (0, T; {\mathbb R})$.
On the other hand, combining Assertion (ii) and the uniqueness of ${\widehat X}^*$ implies that the adjoint equation \eqref{eq:adjoint} admits a unique solution such that
$(p^*, q^*, u^*) \in {\cal S}^2_{\cal F} (0, T; {\mathbb R}) \times {\cal L}^2_{\cal F} (0, T; {\mathbb R}^n)
\times {\cal L}^{2,N}_{\cal F} (0, T; {\mathbb R}^m)$.

Therefore, \eqref{eq:PDE-g} must have a unique solution. Otherwise, suppose that there exist two solutions $g_1$ and $g_2$ to \eqref{eq:PDE-g}. Denote by $Y_i (t) : = 2 e^{2 r (T - t) + g_i (t, \, \lambda (t)) }$, for $i =1, 2$, as
defined in \eqref{eq:Y-ansatz}. Following the derivations in Section 3, we can show that both $p^*_1 := Y_1
{\widehat X}^*$ and $p^*_2 := Y_2 {\widehat X}^*$ are the solutions to the adjoint equation \eqref{eq:adjoint}.
This is a contradiction and violates the uniqueness of $p^*$. This completes the proof of the uniqueness.
\end{proof}

\section{Connection with the HJB equation approach}
\label{app:hjb}

In this appendix, we briefly discuss the connection between the stochastic maximum approach and the HJB equation approach, since both approaches can be applied to solve the quadratic-loss minimization problem \eqref{pro:qlm}.
We refer interested readers to \cite{ait2016portfolio}, \cite{cao2019optimal}, and \cite{liu2021household} for the 
applications of the HJB approach to optimal investment problems under the Hawkes jump models.

Let us start by considering the dynamic version of the quadratic-loss minimization problem \eqref{pro:qlm} at $(t, x, \lam)$, defined as follows:
\begin{align}
	\label{pro:dynamic_qlm}
\Vc(t, x, \lam) := \min_{\pi \in \Ac} \; \Eb \big[ \wx^2(T) \big| \wx(t) = x, \, \lam(t) = \lam \big],
\end{align}
where the dynamics of $\wx$ is given by \eqref{eq:dX_hat}.

Firstly, we write down the HJB equation for problem \eqref{pro:qlm} (or problem \eqref{pro:dynamic_qlm} to be precise):
\begin{align}\label{eq:HJB}
\inf_{\pi \in \Rb} \,  {\cal L}^\pi [\Vc (t, x, \lambda)]  = 0 ,
\end{align}
with the terminal condition $\Vc (T, x, \lambda) = x^2$.
The partial differential operator ${\cal L}^\pi$ acting on any smooth function $\varphi$ and admissible $\pi$
is defined by
\begin{align}
{\cal L}^\pi [\varphi (t, x, \lambda)] : =& \ \varphi_t (t, x, \lambda)
+ \big ( r x + \pi^\top B \big ) \varphi_x (t, x, \lambda)
+ \big ( \alpha \lambda_\infty + (\beta - \alpha) \lambda \big )^\top \varphi_\lambda (t, x, \lambda) \\
& + \frac{1}{2} \pi^\top \sigma \sigma^\top \pi \varphi_{xx} (t, x, \lambda) + \sum^m_{j = 1} \int_{(-1,\infty)}
\big [ \varphi (t, x + (\pi^\top \eta (z))_{(j)}, \lambda + \beta_{(j)}) \\
& - \varphi (t, x, \lambda) - \varphi_x (t, x, \lambda) (\pi^\top \eta (z))_{(j)}
- \beta_{(j)}^\top \varphi_\lam (t, x, \lambda) \big ] \lam_j \nu_j (\dd z_j)
\end{align}
with $(\pi^\top \eta (z))_{(j)}$ being the $j^{\text{th}}$ component of the vector $\pi^\top \eta (z)$ and dependent on
only the $j^{\text{th}}$ coordinate of $z$ (i.e., $z_j$), and $\beta_{(j)}$ being the $j^{\text{th}}$ column vector of the matrix $\beta$.

We try the following {\it ansatz} for the value function $\Vc$ in \eqref{pro:dynamic_qlm}:
\begin{align}
\Vc (t, x, \lambda) = e^{2 r (T-t) + \kappa (t, \lambda)} x^2 ,
\end{align}
with the terminal condition $\kappa (T, \lambda) = 0$.
Substituting this into \eqref{eq:HJB} and after some algebraic manipulation, we obtain
\begin{align}
& \inf_{\pi\in\mathbb R} \bigg \{ 2 x \pi^\top B + \pi^\top \sigma \sigma^\top \pi + \sum^m_{j = 1} \int_{(-1,\infty)}
\big [ (\pi^\top \eta (z))_{(j)}^2 e^{\kappa (t, \lambda + \beta_{(j)})- \kappa} + 2 x (\pi^\top \eta (z))_{(j)} (e^{\kappa (t, \lambda + \beta_{(j)})- \kappa} - 1) \big ] \lam_j \nu_j (\dd z_j) \bigg \} \\
& \qquad + x^2 \bigg\{ \kappa_t +
\kappa_\lambda^\top \big ( \alpha \lambda_\infty + (\beta - \alpha) \lambda \big ) + \sum^m_{j = 1} \int_{(-1,\infty)}
\big [ e^{\kappa (t, \lambda + \beta_{(j)})- \kappa} - 1
- \beta_{(j)}^\top \kappa_\lam  \big ] \lam_j \nu_j (\dd z_j) \bigg \} = 0 ,
\end{align}
where we have suppressed the arguments $(t, \lam)$ for $\kappa$ and its partial derivatives, except the post-jump ones.

The first-order condition to the above minimization problem yields
\begin{align}
2 x \, \wz^\kappa (t, \lambda)  + 2 \Gamma^\kappa (t, \lambda) \pi = {\bf 0_m} ,
\end{align}
where
\begin{align}
 \Gamma^\kappa (t, \lambda) :=& \ \sig \sig^\top + \int_{(-1, \infty)^m} \; \eta(z) \, \mathrm{Diag}[(U^\kappa (t, \lambda) + \bm{1}_m) \bullet \lam \bullet \nu(\dd z)] \, \eta(z)^\top,
\label{eq:Gamma(t,b)_f} \\
\wz^\kappa (t, \lambda) :=& \ B + \int_{(-1, \infty)^m} \; \eta(z) \, \mathrm{Diag}[\lam \bullet \nu(\dd z)] \, U^\kappa (t, \lambda),
\label{eq:Z_hat_f}
\end{align}
are defined similarly as \eqref{eq:Gamma(t,b)} and \eqref{eq:Z_hat} but with $U (t, \lambda)$ in \eqref{eq:U(b)} replaced by
\begin{align}
U^\kappa (t, \lambda) :=& \bigg ( e^{\kappa (t, \, \lambda +\beta_{(1)}) - \kappa (t, \, \lambda)} - 1,
\cdots, e^{\kappa (t, \, \lambda +\beta_{(m)}) - \kappa (t, \, \lambda)} - 1 \bigg )^\top. \label{eq:U(b)_f}
\end{align}
Thus, the optimal strategy can be represented in the following feedback form:
\begin{align}
\pi^*_c (t, x, \lam) = - \Gamma^\kappa (t, \lam)^{-1} \, \wz^\kappa (t, \lam) \; x .
\end{align}

Plugging $\pi^*_c (t, x, \lam)$ into \eqref{eq:HJB} gives
\begin{align}\label{eq:PDE-f}
\kappa_t(t, \lam) + \kappa_\lambda(t,\lam)^\top \, \alpha (\lambda_{\infty} - \lambda )
+ U^\kappa (t, \lambda)^\top \lambda = \wz^\kappa (t, \lambda)^\top \Gamma^\kappa (t, \lambda)^{-1} \wz^\kappa (t, \lambda) .
\end{align}
Indeed, the above equation \eqref{eq:PDE-f} is exactly the same as \eqref{eq:PDE-g}.
It then follows from the existence and uniqueness result in Appendix \ref{app:pde}
that $\kappa (t, \lam) = g (t, \lam)$, for any $(t, \lam) \in [0, T] \times \mathbb R^m_+$. Therefore, the stochastic maximum principle approach and the HJB equation approach lead to the same solution for problem \eqref{pro:qlm}.

Indeed, such a conclusion is guaranteed by the relationship between the two approaches. That is,
the value function and the adjoint process are related as follows:
\begin{align}
 p^* (t) &= 2 e^{2 r (T-t) + g (t, \lambda (t))} {\widehat X}^* (t)
= \Vc_x (t, {\widehat X}^* (t), \lam (t)) , \\
 q^* (t) &= \sigma^\top \pi^*_c (t) Y (t) = \sigma^\top \pi^*_c (t) \Vc_{xx} (t, {\widehat X}^* (t), \lam (t)) , \\
 u^*_l (t, z_l) &= {\widehat X}^* (t-) V_{l} (t)
+ \sum^k_{i = 1}  \pi^*_{ci} (t) \eta_{il} (z_l) \, \big ( Y (t-)
+ V_{l} (t) \big ) \\
&= \Vc_x (t, {\widehat X}^* (t-) + (\pi^*_c (t)^\top \eta (z))_{(l)}, \lam (t) + \beta_{(l)})
- \Vc_x (t, {\widehat X}^* (t-), \lam (t)) .
\end{align}
Moreover, substituting $p^* (t) = \Vc_x (t, {\widehat X}^* (t), \lam (t))$ into the adjoint equation \eqref{eq:adjoint} and matching the drift gives the relationship between the Hamiltonian and the value function as below:
\begin{align}
{\cal L}^{\pi^*_c} [\Vc_x (t, {\widehat X}^* (t), \lam (t))] = - {\cal H}_x (t, {\widehat X}^* (t), \pi^*_c (t), p^* (t), q^* (t), u^* (t)) .
\end{align}
One can refer to Theorem 3.1 in \cite{framstad2004mp} and Theorem 5.6 in \cite{oksendal2005applied} for the relationship between the stochastic maximum principle and dynamic programming principle in general jump-diffusion control systems.


\end{document}